\documentclass[12pt]{article}

\usepackage{float}
\usepackage{authblk}
\usepackage{natbib}
\usepackage{amsmath,amssymb,amstext,amsthm,amsfonts}
\usepackage{mathtools}
\usepackage{placeins}
\usepackage{color}
\usepackage{xcolor}
\usepackage{multirow}
\usepackage{mathrsfs}
\usepackage{listings}
\usepackage{graphicx}
\usepackage[hyphens]{url}
\usepackage{enumitem}
\usepackage{subcaption}
\mathtoolsset{showonlyrefs}

\usepackage{float}

\usepackage[hidelinks]{hyperref}
\hypersetup{
	colorlinks,
	linkcolor={blue},
	citecolor={blue},
	urlcolor={blue}
}

\newtheorem{thm}{Theorem}
\newtheorem{prp}{Proposition}
\newtheorem{lem}{Lemma}

\newtheorem{ass}{Assumption}

\theoremstyle{definition}

\newtheorem{exm}{Example}
\newtheorem{rmk}{Remark}

\newcommand{\single}{\spacingset{1}}

\DeclareMathOperator*{\argmin}{arg\,min}

\def\spacingset#1{%
  \renewcommand{\baselinestretch}{#1}%
  \small\normalsize
}

\title{\bf Fr\'echet Kriging for Metric Space-Valued Spatial Data}
\author[1,2,*]{Ryo Okano}
\author[3,*]{Daisuke Kurisu}
\author[4]{Hans-Georg M\"uller}
\affil[1]{Graduate School of Economics, Hitotsubashi University}
\affil[2]{RIKEN Center for Advanced Intelligence Project}
\affil[3]{Faculty of Economics, The University of Tokyo}
\affil[4]{Department of Statistics and Data Science, University of California, Davis}

\footnotetext[1]{The first two authors contributed equally to this work.}
\footnotetext[2]{R.O. was partially supported by JSPS Grant-in-Aid for JSPS Fellows (25KJ0140). D.K. was partially supported by JSPS KAKENHI Grant Numbers 25K00624 and 26K00323. H.G.M. was partially supported by NSF grant DMS2310450. }

\begin{document}

\spacingset{1}

\maketitle

\begin{abstract}
Spatial prediction of non-Euclidean data, such as probability distributions, 
symmetric positive semidefinite matrices, and networks, arises in many modern 
applications. Such data are naturally modeled as elements of general metric 
spaces, where classical kriging methods that rely on vector space structure 
are no longer directly applicable. We introduce Fr\'echet kriging, a 
general framework for spatial prediction of metric space-valued data. 
Fr\'echet kriging replaces the weighted average in classical ordinary 
kriging with a weighted Fr\'echet mean, while the weights are determined 
by characterizing the spatial dependence structure through a generalized 
variogram. This construction allows the proposed method to preserve the 
intrinsic geometry of the response space without requiring an embedding 
into a linear space. We establish theoretical properties of the proposed predictor, including
a surrogate optimality property and consistency under variogram
misspecification.
Simulation studies 
for univariate distributions and symmetric positive semidefinite matrices 
demonstrate favorable predictive performance compared with existing 
approaches. We further illustrate the practical utility of the proposed 
framework through applications to joint probability densities of daily minimum and maximum 
temperatures and inter-industry monetary flow networks.
\end{abstract}

\noindent
{\it Keywords:} Fr\'echet mean, metric statistics, non-Euclidean data, spatial statistics

\newpage

\newcommand{\double}{\spacingset{1.9}}
\double

\section{Introduction}
Kriging is one of the most widely used methods for spatial prediction and a fundamental tool in spatial statistics \citep{cressie2015statistics, chiles1999geostatistics, stein1999interpolation}. Given observations at a set of spatial locations, kriging predicts the response at an unobserved location by exploiting the spatial dependence among the observations. Originally developed in the context of mining, kriging has since found applications in a wide range of fields, including public health and environmental science \citep{waller2004applied,webster2007geostatistics, wackernagel2013multivariate}. Among its various forms, ordinary kriging is the most common version. It assumes an unknown constant mean and a stationary spatial dependence structure and constructs the best linear unbiased predictor (BLUP) by minimizing the prediction variance subject to an unbiasedness constraint \citep{matheron1963principles}.

Although classical kriging has traditionally been developed for Euclidean responses, some applications involve functional responses. This has motivated extensions of the classical kriging framework to functional data \citep{goulard1993geostatistical, nerini2010cokriging, giraldo2010continuous, giraldo2011ordinary, menafoglio2013universal, mateu2022introduction}. In addition, \cite{menafoglio2014kriging} studied spatial prediction for distributional responses within the Bayes Hilbert space framework. These methods formulate kriging in a Hilbert space, where linear operations such as addition and scalar multiplication are well defined; see %erefer to 
\cite{menafoglio2017statistical} for a review of kriging methods for Hilbert space-valued responses.

In modern applications, one increasingly encounters complex-structured data taking values in a metric space, such as distributions \citep{petersen2022modeling}, symmetric positive semidefinite matrices \citep{dryden2009non}, networks \citep{zhou2022network}, compositional data \citep{scea:11} and set-valued data \citep{kuri:25}. Such data exhibit inherently non-Euclidean structure, making it important to account for the geometry of the response space when constructing spatial predictors. To accommodate these settings, \cite{pigoli2016kriging} proposed a kriging method for manifold-valued data, with particular emphasis on symmetric positive definite matrices. Their method adopts an extrinsic approach, in which the data are mapped to a tangent space via the logarithmic map and kriging is subsequently performed on the resulting representations. Although this provides a convenient linear representation, tangent space coordinates generally do not preserve all pairwise geodesic distances on a curved manifold, so the resulting prediction may depend on the choice of reference point and the curvature of the underlying space; additionally the inverse map is often not defined on the entire tangent space, which means that predictions cannot be mapped back into the object space.

Other approaches have exploited specific structure of the response space. \citet{gouet2015geodesic} developed geodesic kriging for responses in Hadamard spaces, with particular attention to the univariate Wasserstein
space. \citet{balzanella2020spatial}, \citet{gossard2025kriging}, and \citet{qiu2026spatial} also developed kriging methods for univariate distributions equipped with the Wasserstein metric. These methods leverage geometric or linear structures available in Hadamard or univariate Wasserstein spaces and therefore do not provide a unified formulation for arbitrary metric space-valued responses. Network-valued kriging was also briefly discussed by \citet{severn2021nonparametric}, who mentioned a linear kriging predictor for graph Laplacians with weights selected by minimizing the mean squared prediction error as a future research project without providing further details. 

In this paper, we propose \emph{Fr\'echet kriging}, a generalization of ordinary kriging for spatial prediction with metric space-valued data. The key idea is to replace the weighted linear combination of responses in ordinary kriging with a weighted Fr\'echet mean \citep{Frechet1948}. The weights are determined by a generalized variogram that characterizes spatial dependence through pairwise distances between responses. Both the characterization of spatial dependence and the construction of the predictor therefore rely only on the metric structure of the response space. This yields a unified framework for spatial prediction that does not require linear operations, embeddings, or geometric structures specific to particular response spaces. In particular, Fr\'echet kriging operates directly on the original metric space, thereby accommodating a broad range of complex responses while respecting their intrinsic metric structure. Examples include multivariate probability distributions and networks, for which existing kriging methodology remains relatively underdeveloped. We establish theoretical properties of the proposed predictor, including a surrogate optimality property and robustness to variogram misspecification. In particular, we show that the proposed predictor remains consistent, in the sense that its prediction risk converges to zero, even when the working variogram used to construct the prediction weights is misspecified.

The use of weighted Fr\'echet means for spatial interpolation has been considered in earlier work. In particular, \citet{dryden2009non} considered interpolation between covariance matrices on a grid using weighted Fr\'echet means and mentioned the possibility of using kriging-based weights. Similarly, \citet[Chapter~1]{halvorsen2012kriging} mentions that one could  extend linear kriging to covariance matrix-valued data by replacing the weighted linear combination with a weighted Fr\'echet mean. However, these works do not develop a general metric-space kriging framework with an explicit variogram-based weight construction and theoretical guarantees for spatial prediction. Here we %Our contribution is to 
develop such a framework, together with a concrete weight-selection procedure and supporting theory. %establishing its prediction properties.

A complementary line of research has developed regression methods for responses taking values in general metric spaces. These include global and local Fr\'echet regression \citep{petersen2019frechet}, as well as subsequent developments involving sufficient dimension reduction \citep{ying2022frechet, zhang2024dimension}, single-index models \citep{bhattacharjee2023single, ghosal2023frechet}, and deep learning \citep{iao2025deep, zhou2026end}.
 These methods provide a general framework for predicting metric space-valued responses from Euclidean covariates. However, they are primarily developed for independent observations and do not  account for dependence among responses. In spatial settings, such dependence conveys critical information for prediction and must therefore be incorporated. % into the %predictions. procedure. 
 Fr\'echet kriging does this and explicitly exploits the dependence structure among metric space-valued responses to construct spatial predictions. 

The remainder of the paper is organized as follows.
Section~\ref{sec:method} introduces the proposed method,
and Section~\ref{sec:theory} establishes its theoretical properties.
Sections~\ref{sec:simulation} and \ref{sec:real_data}
present simulation studies and real data applications,
respectively, with conclusions in 
Section~\ref{sec:conc_rmk}. 
Additional results and proofs are provided in the
Supplement.

\section{Methodology}\label{sec:method}
\subsection{Setting}
Let $(\Omega,d)$ be a metric space. 
Suppose we observe pairs $(x_i,Y_i)$, $i=1,\ldots,n$, where
$x_i\in\mathbb{R}^p$ is a deterministic spatial location and
$Y_i\in\Omega$ is the corresponding response.
Throughout, we assume that the sampling locations
$x_1,\ldots,x_n$ are pairwise distinct and below in Assumption 5 we will require an infill/fixed-domain asymptotic condition. 
We further assume that the observations are generated according to the model
\begin{equation}
    Y_i=Y(x_i),\quad Y(x) = \mathcal{P}_x(S(x)), \quad x \in \mathbb{R}^p.
    \label{eq:basic_model}
\end{equation}
Here, $\{S(x):x\in\mathbb R^p\}$ is a latent $\Omega$-valued stochastic
process, and for each $x\in\mathbb R^p$
$\mathcal{P}_x:\operatorname{supp}\{S(x)\}\to\Omega$ is a random
perturbation map, where $\operatorname{supp}\{S(x)\}$ denotes the support
of the distribution of $S(x)$, such that
\begin{equation}
    \omega
    =
    \arg\min_{z\in\Omega}
    E\left[
        d^2\{\mathcal{P}_x(\omega),z\}
    \right] \quad \text{for every $\omega \in \operatorname{supp}\{S(x)\}$},
    \label{eq:pertmat_mean}
\end{equation}
and
\begin{equation}
    E\left[
        d^2\{\mathcal{P}_x(\omega),\omega\}
    \right]
    =
    \tau^2 \quad \text{for every $\omega \in \operatorname{supp}\{S(x)\}$}.
    \label{eq:pertmat_var}
\end{equation}
This perturbation map generalizes  additive errors $\varepsilon(x)$ that are commonly used in   Euclidean settings, where one requires  $E[\varepsilon(x)]=0$ and
$\operatorname{Var}\{\varepsilon(x)\}=\tau^2$ \citep{chen2022uniform}.
We assume that the perturbation maps $\{\mathcal{P}_x:x\in \mathbb{R}^p\}$ are mutually independent and independent of the latent process $S$.
In the classical case of kriging with scalar data the   parameter $\tau^2$ is referred to as the \emph{nugget variance} \citep{cressie2015statistics}.

Analogous to classical ordinary kriging, we impose a constant mean assumption on the latent process $S$ and the observation process $Y$. For each $x \in \mathbb{R}^p$, the Fr\'echet means \citep{Frechet1948} of $S(x)$ and $Y(x)$ are defined as  
\begin{equation}
    \mu_{S}(x) = \argmin_{\omega \in \Omega} E[d^2(S(x), \omega)], 
    \quad 
    \mu_{Y}(x) = \argmin_{\omega \in \Omega} E[d^2(Y(x), \omega)].
\end{equation}

\begin{ass}\label{ass:const_mean}
    The processes $S$ and $Y$ have constant Fr\'echet means $\mu_S(x) = \mu_S$ and $\mu_Y(x) = \mu_Y$, respectively.
\end{ass}

Furthermore, analogous to classical ordinary kriging, we also assume that the processes have constant variances and stationary and isotropic spatial dependence structures. Under  Assumption \ref{ass:const_mean}, the Fr\'echet variances \citep{Frechet1948} of $S(x)$ and $Y(x)$ are defined as 
\begin{equation}
    \sigma^2_{S}(x) = E[d^2(S(x), \mu_{S})], \quad 
    \sigma^2_{Y}(x) = E[d^2(Y(x), \mu_{Y})].
\end{equation}
To characterize the spatial dependence structure, we consider the variogram \citep{cressie2015statistics}.
 When $\Omega = \mathbb{R}$ and Assumption \ref{ass:const_mean} holds, the variogram of $S$ is the function $V_S(x, x') = E[d_E^2(S(x), S(x'))]/2$, where $d_E$ is the Euclidean distance. Motivated by this representation, we define the variograms of $S$ and $Y$ in the general metric-space setting as
\begin{equation}
    V_S(x, x') = \frac{1}{2}E[d^2(S(x), S(x'))], \quad 
    V_Y(x, x') = \frac{1}{2}E[d^2(Y(x), Y(x'))].
    \label{eq:vario_sy}
\end{equation}

\begin{ass}\label{ass:second_stationary}
    The processes $S$ and $Y$ have constant Fr\'echet variances $\sigma_S^2(x) = \sigma_S^2$ and $\sigma_Y^2(x) = \sigma_Y^2$, respectively. In addition, their variograms are of the form  $V_S(x, x') = \varphi_S(\|x-x'\|)$ and  $V_Y(x, x') = \varphi_Y(\|x-x'\|)$, where $\varphi_S$ and $\varphi_Y$ are functions from $[0, \infty)$ to $[0, \infty)$. 
\end{ass}

We impose Assumptions~\ref{ass:const_mean} and \ref{ass:second_stationary} throughout the paper. In Section~\ref{sec:simulation}, we provide several constructions of $\Omega$-valued stochastic processes $\{S(x): x \in \mathbb{R}^p\}$ and perturbation maps $\{\mathcal{P}_x: x \in \mathbb{R}^p\}$ for which the resulting latent and observation processes satisfy these assumptions.
Note that $\varphi_S(0) = \varphi_Y(0) = 0$ by the definition of the variogram. 

Before introducing the proposed method, we present several examples of metric spaces $(\Omega, d)$ that fit within the proposed framework. In addition to standard Euclidean spaces, $\Omega$ may also be a Hilbert space. In particular, functional data defined on a compact domain $\mathcal T$ can be regarded as elements of the Hilbert space $L^2(\mathcal T)$ of square-integrable functions with the $L^2$-metric $d_{L^2}(f,g)=(\int_\mathcal{T} \{f(t)-g(t)\}^2dt)^{1/2}$ and therefore fall within our framework.
Additional examples are provided below.

\begin{exm}[Univariate distributions]\label{exm:unidist}
Let $I$ be a closed interval, and let $\mathcal{P}_2(I)$ denote the space of probability distributions on $I$ with finite second moments. A common metric on $\mathcal{P}_2(I)$ is the 2-Wasserstein metric \citep{panaretos2020invitation}, defined by
\begin{equation}
d_W(\nu_1, \nu_2)
=
\left(
\int_0^1
\{F_1^{-1}(u)-F_2^{-1}(u)\}^2du
\right)^{1/2},
\end{equation}
where $F_1^{-1}$ and $F_2^{-1}$ denote the quantile functions of $\nu_1$ and $\nu_2$.
The resulting metric space is called the 2-Wasserstein space of univariate distributions. 
\end{exm}

\begin{exm}[Positive semidefinite matrices]\label{exm:spd}
    Consider the space $\mathrm{Sym}_{q}^{+}$ of $q \times q$ symmetric positive semidefinite (SPSD) matrices.
    This space can be equipped with the Frobenius metric,  the power metric \citep{dryden2009non} and the Bures--Wasserstein metric \citep{bhatia2019bures}. The power metric with exponent $1/2$ is also known as the square-root metric.

    % the power metric \citep{dryden2009non},  the Log-Euclidean metric \citep{arsigny2007geometric} and the Log-Cholesky metric \citep{lin2019riemannian}. 
    % Let $\mathrm{Sym}_{q}^{+}$ denote the space of $q \times q$ symmetric positive semidefinite (SPSD) matrices. 
    % This space can be equipped with the power metric. 
\end{exm}

\begin{exm}[Multivariate distributions]\label{exm:multdist}
Let $D \subset \mathbb{R}^k$, and let $\mathcal{P}(D)$ denote the space of probability distributions on $D$ that have density functions. A common metric on $\mathcal{P}(D)$ is the Fisher--Rao metric \citep{zhu2024spherical}, defined as the geodesic distance between the square-root density functions:
\begin{equation}
d_{\mathrm{FR}}(\nu_1,\nu_2)
=
\arccos \left\{
\int_D \sqrt{f_1(x)f_2(x)}dx
\right\},
\end{equation}
where $f_1$ and $f_2$ denote the density functions of $\nu_1$ and $\nu_2$.
\end{exm}

\begin{exm}[Networks]\label{exm:network}
Consider the space of simple, undirected, weighted networks with a fixed number of nodes and bounded edge weights. Such a network can be uniquely represented by its graph Laplacian matrix. Accordingly, the network space can be identified with the corresponding space of graph Laplacians, which can be equipped with the Frobenius metric \citep{kola:14, severn2022manifold, zhou2022network}.
\end{exm}

\subsection{Fr\'echet Kriging}
Let $x_0 \in \mathbb{R}^p$ be a new location. Our goal is to predict the value of the latent process $S(x_0)$ based on the observations $(x_i, Y_i), i=1, \ldots, n$ generated from the model \eqref{eq:basic_model}. 

For now, assume that the  variogram function $\varphi_Y$ of the observation process $Y$ is known.
Before introducing the proposed method for general metric space-valued data, we first revisit the classical Euclidean case, which serves as the starting point for our generalization.
Suppose $\Omega = \mathbb{R}$ and the noise is additive, that is, $Y(x) = S(x) + \varepsilon(x)$, where $E[\varepsilon(x)] = 0$ and $\text{Var}(\varepsilon(x)) = \tau^2$. In this case, the ordinary kriging predictor of $S(x_0)$ is defined as a weighted average of the responses
\begin{equation}
    \tilde{S}(x_0)
    =
    \sum_{i=1}^n \lambda_i(x_0)Y_i, 
\end{equation}
where the weights $(\lambda_i(x_0))_{i=1}^n$ are the solution to 
\begin{equation}
    \min_{\lambda \in \mathbb{R}^n} E\left[\left(\sum_{i=1}^n\lambda_i Y_i - S(x_0)\right)^2\right] \quad \text{subject to} 
    \quad 
    \sum_{i=1}^n \lambda_i = 1.
\end{equation}
 In other words, the predictor $\tilde{S}(x_0)$ is the best linear unbiased predictor (BLUP) of $S(x_0)$; the data-adaptive choice of the weights $\lambda_i$ is discussed below.
 
It can be shown that the vector of weights $\lambda(x_0) = (\lambda_i(x_0))_{i=1}^n$ can be expressed solely in terms of the variogram function $\varphi_Y$ as
\begin{equation}
\lambda(x_0)
=
\begin{pmatrix}
I_n & \mathbf{0}_n
\end{pmatrix}
\begin{pmatrix}
\Gamma & \mathbf{1}_n \\
\mathbf{1}_n^\top & 0
\end{pmatrix}^{-1}
\begin{pmatrix}
\gamma(x_0) \\
1
\end{pmatrix},
\label{eq:weights_vario}
\end{equation}
where $\Gamma$ is an $n \times n$ matrix whose $(i, j)$ element is given by $\varphi_Y(\|x_i - x_j\|)$, 
$\gamma(x_0)$ is an $n$-dimensional vector whose $i$-th element is given by $\varphi_Y(\|x_i - x_0\|)$, and $\mathbf{1}_n$ and $\mathbf{0}_n$ are 
the $n$-dimensional vectors of ones and zeros, respectively, and 
$I_n$ is the $n \times n$ identity matrix  \citep[Chapter~3]{cressie2015statistics}. 
Throughout, we assume that the block matrix
appearing in \eqref{eq:weights_vario} is nonsingular, so that the inverse in
\eqref{eq:weights_vario} exists.

Since the weights satisfy $\sum_{i=1}^n \lambda_i(x_0)=1$, 
the ordinary kriging predictor can equivalently be characterized as
\begin{equation}
    \tilde{S}(x_0) = \argmin_{y \in \mathbb{R}}\sum_{i=1}^n \lambda_i(x_0) d_E^2(Y_i, y),
\end{equation}
where $d_E$ is the Euclidean metric.
This motivates us to define a predictor for a general metric space $\Omega$ as 
\begin{equation}
    \tilde{S}(x_0) = \argmin_{\omega \in \Omega} \sum_{i=1}^n \lambda_i(x_0) d^2(Y_i, \omega),
    \label{eq:fk}
\end{equation}
where the vector of weights $\lambda(x_0)$ is given by \eqref{eq:weights_vario}. 
Note that for the above  \emph{Fr\'echet kriging} predictor both the variogram used to 
determine the kriging weights and the weighted least-squares 
characterization of the predictor are generalized using the metric $d$.

In practice, the variogram function $\varphi_Y$ is unknown and must be estimated from the data. Let $\hat{\varphi}_Y$ denote an estimator of $\varphi_Y$. The specific estimation procedure is described in Section~\ref{sec:implement}. Let $\hat{\lambda}(x_0) = (\hat{\lambda}_i(x_0))_{i=1}^n$ denote the weight vector obtained from \eqref{eq:weights_vario} by replacing $\varphi_Y$ with the estimated variogram $\hat{\varphi}_Y$. The Fr\'echet kriging predictor based on the estimated variogram is then defined as
\begin{equation}
    \hat{S}(x_0)
    =
    \argmin_{\omega \in \Omega}
    \sum_{i=1}^n
    \hat{\lambda}_i(x_0)
    d^2(Y_i, \omega).
    \label{eq:fk_vario_est}
\end{equation}

Since the kriging weights $\lambda(x_0)$ and $\hat{\lambda}(x_0)$ are not necessarily nonnegative, the existence and uniqueness of the minimizers in \eqref{eq:fk} and \eqref{eq:fk_vario_est} are not guaranteed in general; throughout  we assume existence and uniqueness. %that the minimizers in \eqref{eq:fk} and \eqref{eq:fk_vario_est} exist and are unique. 
Under Assumptions \ref{ass:isometric} and \ref{ass:convex} introduced in Section \ref{sec:theory}, these conditions are automatically satisfied, as shown in Lemma \ref{lem:alt_rep} in the Supplement.

\begin{rmk}[Connection to kriging for manifold-valued responses]
\label{rmk:manifold}
\citet{pigoli2016kriging} proposed a kriging approach for manifold-valued 
responses, in which the responses are mapped to a tangent space via the 
logarithmic map and kriging is then performed in the tangent space. 
\citet{qiu2026spatial} adopted a similar approach for distributional
responses in the Wasserstein space and additionally performed dimension
reduction in the tangent space.
Although the proposed Fr\'echet kriging generally differs from the
approach of \citet{pigoli2016kriging}, the two approaches coincide when
the logarithmic map preserves distances and the predictor in the tangent space
is modified by metric projection onto the image of the response space
(see Section~\ref{sec:kriging_manifold} of the Supplement
for details).
% The distance-preserving property of the logarithmic map holds, for
% example, when $\Omega$ is the space of SPD matrices equipped with the
% Log-Euclidean metric, the space of SPSD matrices equipped with the
% square-root metric \citep[Section~2]{pigoli2016kriging}, or the
% Wasserstein space of univariate distributions
% \citep{panaretos2020invitation}.
\end{rmk}

% \begin{rmk}[Comparison with distribution-valued kriging by \citet{gossard2025kriging}]
% \label{rmk:dist}
% Let $(\Omega,d)$ be the $2$-Wasserstein space of univariate distributions
% (Example~\ref{exm:unidist}), and let $Q_\mu$ denote the quantile function
% of a distribution $\mu\in\Omega$.
% \citet{gossard2025kriging} proposed a kriging predictor for
% $Q_{S(x_0)}$ of the form $\tilde{Q}_{S(x_0)}
%     =
%     \sum_{i=1}^n
%     \overline{\lambda}_i(x_0) Q_{Y_i}$. 
% Here, the weights $(\overline{\lambda}_i(x_0))_{i=1}^n$ are determined by
% the criterion
% \begin{equation}
%     \min_{\lambda \in \mathbb{R}^n}
%     E\left[
%         \int_0^1
%         \left\{
%             \sum_{i=1}^n \lambda_i Q_{Y_i}(u)
%             -
%             Q_{S(x_0)}(u)
%         \right\}^2
%         du
%     \right]
%     \quad \text{subject to} 
%     \quad 
%     \sum_{i=1}^n \lambda_i = 1.
% \end{equation}
% Since $\tilde{Q}_{S(x_0)}$ is not necessarily nondecreasing,
% they apply a sorting procedure to enforce monotonicity and use the
% resulting function as the predicted quantile function.
% As shown in Section~\ref{sec:rmk_dist}, the Fr\'echet kriging predictor
% and the predictor of \citet{gossard2025kriging} coincide in the first
% step: they use the same kriging weights to construct the same weighted
% average of the quantile functions.
% However, they differ in how this initial predictor is subsequently
% modified. Our method uses a metric projection onto the closed convex
% set of quantile functions, whereas \citet{gossard2025kriging} use a
% sorting procedure.
% \end{rmk}

\section{Theory}\label{sec:theory}
\subsection{Optimality}
We first establish an optimality property of the Fr\'echet kriging predictor \eqref{eq:fk} constructed using the true variogram. This result characterizes the oracle kriging weights as those minimizing a surrogate upper bound on the prediction risk.
 Define the set of weights by
\[
\Lambda_n
=
\left\{\lambda \in \mathbb{R}^n: \sum_{i=1}^n \lambda_i = 1\right\}.
\]
For each $\lambda \in \Lambda_n$, define the corresponding weighted 
Fr\'echet mean predictor of $S(x_0)$ by
\begin{equation}
    \tilde{S}(x_0; \lambda)
=
\argmin_{\omega \in \Omega} \sum_{i=1}^n \lambda_i d^2(Y_i, \omega),
\label{eq:general_predictor}
\end{equation}
and define its prediction risk by 
\begin{equation}
    R(x_0;\lambda) = E[d^2\{{S}(x_0),\tilde{S}(x_0;\lambda)\}].
\end{equation}

We  impose the following assumption on the metric space $(\Omega, d)$. 

\begin{ass}\label{ass:isometric}
    The metric space $(\Omega, d)$ is such that $d^2$ is of negative type. That is, for any $m \ge1$, $z_1, \ldots, z_m \in \Omega$ and $a_1, \ldots, a_m \in \mathbb{R}$ with $\sum_{i=1}^ma_i = 0$, we have \\ 
    $
        \sum_{i=1}^m\sum_{j=1}^m a_i a_j d^2(z_i, z_j) \le 0.
    $
\end{ass}

Under Assumption \ref{ass:isometric}, by a classical result of \cite{scho:38}, there exists a Hilbert space $\mathcal{H}$ equipped with an inner product
$\langle \cdot, \cdot \rangle_{\mathcal{H}}$, the induced norm
$\|\cdot\|_{\mathcal{H}}$, and an isometric embedding
$\Psi:\Omega\to\mathcal{H}$ such that
$
d(x,y)=\|\Psi(x)-\Psi(y)\|_{\mathcal{H}}
$
for all $x,y\in\Omega$. We further assume the following.

\begin{ass}\label{ass:convex}
    The set $\Psi(\Omega)$ is nonempty, closed, and convex in $\mathcal{H}$.
\end{ass}

Assumption \ref{ass:convex} ensures the existence and uniqueness of projections onto ${\Psi}(\Omega)$ \citep[Chapter 3]{deutsch2001best}. This property is used in the optimality analysis below.

Assumptions \ref{ass:isometric} and \ref{ass:convex} are satisfied by many metric spaces arising in applications. Examples include the 2-Wasserstein space of univariate probability distributions (Example \ref{exm:unidist}),  the space of SPSD matrices equipped with the Frobenius and power metrics (Example \ref{exm:spd}), and the space of graph Laplacians equipped with the Frobenius metric (Example \ref{exm:network}).
In contrast, the Fisher--Rao metric for multivariate probability densities (Example~\ref{exm:multdist}) does not satisfy Assumptions~\ref{ass:isometric} and \ref{ass:convex} and hence Proposition~\ref{prp:opt} does not apply. Nevertheless, the corresponding setting is covered by the consistency result in Theorem~\ref{thm:miss} in the next subsection, provided Assumption~\ref{ass:miss} stated there holds.
Under these assumptions, we obtain:

\begin{prp}\label{prp:opt}
Suppose that Assumptions \ref{ass:const_mean}--\ref{ass:convex} hold, and $x_0 \notin\{x_1, \ldots, x_n\}$.
\begin{itemize}
\item[(i)] For any $\lambda \in \Lambda_n$,
\begin{equation}
R(x_0;\lambda)
\le
E\left[
\left\|
\sum_{i=1}^n
\lambda_i \Psi(Y_i)
-
\Psi(S(x_0))
\right\|_{\mathcal H}^2
\right].
\label{eq:risk_bound_surro}
\end{equation}

\item[(ii)] The upper bound in \eqref{eq:risk_bound_surro} is minimized when $\lambda$ is given by \eqref{eq:weights_vario}.
\end{itemize}
\end{prp}

Proposition~\ref{prp:opt} states that the prediction risk is bounded above by the surrogate risk
$
E[
\left\|
\sum_{i=1}^n
\lambda_i \Psi(Y_i)
-
\Psi(S(x_0))
\|_{\mathcal H}^2
\right],
$
and  that the surrogate risk is minimized when $\lambda$ is chosen to be the Fr\'echet kriging weights \eqref{eq:weights_vario}.
Consequently, the Fr\'echet kriging weights are optimal with respect to a Hilbert-space-based surrogate upper bound on the prediction risk.
Since the inequality in \eqref{eq:risk_bound_surro} becomes an equality when $\Omega$ is the Euclidean space $\mathbb{R}$, Proposition~\ref{prp:opt} recovers classical BLUP optimality for the Euclidean case and otherwise establishes surrogate-risk optimality under  Assumptions \ref{ass:const_mean}--\ref{ass:convex}.

\subsection{Convergence under Variogram Misspecification}
We next analyze the convergence of the prediction risk of the Fr\'echet kriging predictor.
To reflect practical settings, we allow the variogram used to construct the predictor
to be misspecified.
Specifically, let $W(x,y)=\psi(\|x-y\|)$ be a fixed working variogram,
where $\psi:[0,\infty)\to[0,\infty)$.
The working variogram function $\psi$ is not required to coincide with the true variogram function
$\varphi_Y$ of the observation process.
Furthermore, let $\lambda_\psi(x_0)=(\lambda_{\psi,i}(x_0))_{i=1}^n$ denote the weight vector obtained from
\eqref{eq:weights_vario} by replacing $\varphi_Y$ with $\psi$.
The corresponding Fr\'echet kriging predictor is then defined as
\begin{equation}
    \hat{S}_\psi(x_0)
    =
    \arg\min_{\omega\in\Omega}
    \sum_{i=1}^n
    \lambda_{\psi,i}(x_0)d^2(Y_i,\omega).
    \label{eq:fk_miss}
\end{equation}

In the following, we show that under suitable regularity conditions the prediction risk
$
    E[
        d^2\{\hat{S}_\psi(x_0),S(x_0)\}
    ]
$
converges to zero as $n\to\infty$, even when the working variogram is misspecified, which is a form of robustness % Robustness of kriging prediction to covariance or variogram misspecification 
that  has been studied extensively in the classical Euclidean setting \citep{yakowitz1985comparison, stein1988asymptotically, stein1990uniform, wang2020prediction}.
We thus provide a metric-space analogue of this robustness phenomenon.
For this,  we impose the following assumptions. 
Let $A\subset\mathbb{R}^p$ be a compact spatial domain containing all
sampling locations $x_1,x_2,\ldots$ and the prediction location $x_0$.
In the asymptotic analysis, the perturbation maps are allowed to depend
on $n$. More precisely, conditions~\eqref{eq:pertmat_mean} and
\eqref{eq:pertmat_var} are assumed with $\mathcal P_x$ replaced by
$\mathcal P_{n,x}$ and $\tau^2$ replaced by $\tau_n^2$.
For notational simplicity, we suppress the index $n$ below.

\begin{ass}\label{ass:miss}\quad 
\begin{enumerate}
\item [(i)] The metric space $(\Omega, d)$ is bounded.
\item [(ii)]
$ \min_{1 \le i \le 
   n}\|x_0-x_i\| \to 0.
$
\item [(iii)]The nugget variance in \eqref{eq:pertmat_var} satisfies
    $
        \tau^2 = \tau_n^2 \to 0
    $
    as $n \to \infty$.
    \item[(iv)] The fixed working variogram admits the spectral representation
\begin{equation}
    W(x,y)
    =
    \int_{\mathbb R^p}
    \left[
        1-\cos\{\xi^\top(x-y)\}
    \right]
    \nu(d\xi),
\end{equation}
where $\nu$ is a nonnegative, radial Borel measure satisfying
\begin{equation}
    \int_{\mathbb R^p}
    (1\wedge\|\xi\|^2)\,\nu(d\xi)
    <\infty.
\end{equation}
In addition, the support of $\nu$ contains a nonempty open subset of
$\mathbb R^p$.

    \item[(v)]
    The working kriging weights are uniformly stable in $\ell_1$:
    \[
        \sup_{n\ge1}
        \sum_{i=1}^n
        |\lambda_{\psi,i}(x_0)|
        <\infty.
    \]

    \item[(vi)]
    With probability one, the sample path
    $
        x\mapsto S(x)
    $
    is continuous on $A$ with respect to the metric $d$.
\end{enumerate}
\end{ass}

Condition (i) imposes a mild regularity requirement on the metric space. The stronger assumption of total boundedness is often imposed in the analysis of metric space-valued data; see, for example, \citet{petersen2019frechet, iao2025deep, dubey2024metric}.
Condition (ii) requires that at least one sampling location approaches the fixed prediction location \citep{yakowitz1985comparison}.
Condition (iii)
is a small error condition requiring the 
errors induced by the perturbation maps to be asymptotically negligible. 
Classical robustness results for kriging under covariance
misspecification typically are for  settings without a nugget effect, that is, $\tau_n=0$ \citep{yakowitz1985comparison,stein1988asymptotically,stein1990uniform}.
In contrast, our framework allows for a nonzero nugget effect in finite samples, provided that it vanishes asymptotically. 
Similar small error assumptions are commonly used in nonlinear models with 
measurement errors \citep{carroll1995measurement,schennach2016recent} and 
have also been employed in mixed-effects modeling for metric space-valued 
data \citep{bhattacharjee2025geodesic}.
Condition (iv) imposes a spectral regularity condition on the working
variogram. This condition is satisfied by several commonly used
variogram models, including the nugget-free Mat\'ern variogram,
the nugget-free cardinal sine variogram when $p=2$,
the power variogram with exponent $\alpha\in(0,2)$, and
the powered-exponential variogram with exponent $\alpha\in(0,2]$;
see \citet{cressie2015statistics} for general references.
Condition (v) requires the resulting kriging weights to remain uniformly
bounded in the $\ell_1$ norm. This condition is analogous to the bounded Lebesgue constant condition
used by \citet{vazquez2009pointwise} to establish pointwise consistency
of kriging predictors. Condition (vi) assumes continuity of
the sample paths of the latent process.

Under these assumptions, we obtain the following convergence result.

\begin{thm}\label{thm:miss}
Suppose Assumptions
\ref{ass:const_mean}, \ref{ass:second_stationary} and 
\ref{ass:miss} hold.
Then  the estimator \eqref{eq:fk_miss} with the working variogram satisfies 
%\begin{equation}
    $E\left[
    d^2\{\hat{S}_\psi(x_0),S(x_0)\}
    \right]
    \to 0$ as $n \to \infty$.
%\end{equation}
\end{thm}

Theorem~\ref{thm:miss} establishes risk consistency
even when the working variogram is misspecified. 
%provided that Assumption~\ref{ass:miss} holds.
This extends classical robustness results
for kriging to metric space-valued responses
\citep{yakowitz1985comparison,stein1988asymptotically,
stein1990uniform,wang2020prediction}.

\begin{rmk}[Consistency without boundedness]\label{rmk:without_bounded}
To prove Theorem~\ref{thm:miss}, we assume that the metric space $(\Omega,d)$ is bounded. However, when $(\Omega,d)$ admits an isometric embedding into a Hilbert space, the boundedness assumption can be relaxed. Specifically, suppose that Assumptions~\ref{ass:const_mean}--\ref{ass:convex} and Assumption~\ref{ass:miss}(ii)--(vi) hold. If
$
E\left[\sup_{x\in A}d^2\{S(x),S(x_0)\}\right]<\infty,
$
then
\begin{equation}
E\left[
d^2\{\hat{S}_\psi(x_0),S(x_0)\}
\right]
\to0,
\qquad\text{as } n\to\infty.
\label{eq:converge_without_bounded}
\end{equation}
The proof is provided in Section 
\ref{sec:without_bounded} in the Supplement.
\end{rmk}

\begin{rmk}[Uniform convergence of the prediction risk]\label{rmk:unif}
The pointwise convergence in Theorem~\ref{thm:miss} can be strengthened
to uniform convergence over the spatial domain $A$ under stronger assumptions. 
Specifically, suppose that Assumption~\ref{ass:miss}(ii) is strengthened to
\begin{equation}
    \sup_{x_0\in A}
    \min_{1\le i\le n}
    \|x_0-x_i\|
    \to0,
    \label{eq:dens_sample}
\end{equation}
and that the $\ell_1$ stability condition holds uniformly over $x_0\in A$:
\begin{equation}
    \sup_{n\ge1}\sup_{x_0\in A}
    \sum_{i=1}^n
    |\lambda_{\psi,i}(x_0)|
    <\infty.
\end{equation}
Then the argument used in the proof of Theorem~\ref{thm:miss} yields
\begin{equation}
    \sup_{x_0\in A}
    E\left[
        d^2\{\hat{S}_\psi(x_0),S(x_0)\}
    \right]
    \to0.
\end{equation}
See Section \ref{sec:rmk_unif} of the Supplement for details.
Uniform convergence is particularly relevant for spatial prediction,
where the goal is to construct a predicted spatial field over an entire
spatial domain rather than to predict at a single location.
Condition~\eqref{eq:dens_sample} is satisfied, for example, when the
sampling locations form an increasingly dense regular design over the
compact domain $A$.
\end{rmk}

We emphasize that Theorem~\ref{thm:miss} assumes that the response space is bounded, whereas Remark~\ref{rmk:without_bounded} removes this requirement for metric spaces that admit an isometric embedding into a Hilbert space, subject to an additional moment condition. Both consistency results require a vanishing nugget effect, a fixed working variogram satisfying the regularity condition in Assumption~\ref{ass:miss}(iv), and stable kriging weights.
In contrast, our simulations and real-data applications include settings in which these theoretical conditions are not necessarily satisfied. In particular, we consider nonvanishing nugget effects and variograms estimated from the data. Thus, the simulation studies should be viewed as empirical evaluations of the proposed method beyond the sufficient conditions imposed in the theoretical analysis, rather than as direct illustrations of Theorem~\ref{thm:miss}.

\section{Implementation and Simulation}
\label{sec:simulation}
\subsection{Implementation Details}\label{sec:implement}
To implement the proposed kriging method, we need to estimate the variogram function
$\varphi_Y$ defined in Assumption~\ref{ass:second_stationary}.
Suppose that $\Omega=\mathbb{R}$ and the noise is additive, i.e.,
$Y(x)=S(x)+\varepsilon(x)$, where
$E[\varepsilon(x)]=0$ and
$\operatorname{Var}(\varepsilon(x))=\tau^2$.
In this case, the variogram function is represented as
\begin{equation}
    \varphi_Y(u) = \tau^2 + \sigma_S^2\{1-\rho(u)\},
    \label{eq:expression_vario_S_2}
\end{equation}
for $u > 0$, where $\rho:[0,\infty)\to[-1,1]$ is a function such that 
$\text{Corr}(S(x), S(x')) = \rho(\|x-x'\|)$.
A classical approach to estimating $\varphi_Y$ is to first construct a raw variogram estimate and then fit a parametric model to it using the form on the right-hand side of \eqref{eq:expression_vario_S_2}. In this representation, the intercept $\tau^2$ corresponds to the nugget variance, and the sum $\tau^2 + \sigma_S^2$ corresponds to the variance of the observation process $Y$, sometimes called \emph{sill} \citep{cressie2015statistics}. Motivated by this representation, we adopt the same parametric form as a working variogram in general metric spaces. Outside settings satisfying Assumptions \ref{ass:isometric} and \ref{ass:convex}, however, the fitted parameters need not admit the same variance-component interpretation.

\paragraph{Construction of a raw variogram estimate.}
For each pair $j>i$, compute the distance $u_{ij}=\|x_i-x_j\|$ and its corresponding variogram ordinate $v_{ij} = d^2(Y_i, Y_j)/2$. The pairs $\{(u_{ij}, v_{ij}): j > i\}$ are called the empirical variogram. 
We construct a raw estimate of $\varphi_Y$ based on the pairs in the following way.

Let $K$ denote the number of bins and let $q_{0.8}$ denote the
0.8 quantile of the pairwise distances $\{u_{ij}:j>i\}$.
Given $K$, we set the bin width to $h=q_{0.8}/K$ and define
$
I_{h,k}=((k-1)h,kh],  k=1,\ldots,K.
$
For each $k$, we compute $\hat{\varphi}_k$ as the average of all $v_{ij}$
whose corresponding distances $u_{ij}$ belong to $I_{h,k}$, and regard
$\hat{\varphi}_k$ as an estimate of $\varphi_Y(u_k)$, where
$u_k=(k-0.5)h$ is the midpoint of $I_{h,k}$.
Throughout the simulation and real data analyses, we 
restrict variogram estimation to pairwise distances up to $q_{0.8}$.
This restriction avoids sparsely populated bins at large distances.

\paragraph{Fitting a parametric model.}
We obtain the final estimate of the variogram function $\varphi_Y$ by fitting a
parametric model based on the right-hand side of
\eqref{eq:expression_vario_S_2} to the raw estimates
$\{(u_k,\hat{\varphi}_k)\}_{k=1}^K$.
Let $\rho(u;\theta)$ be a parametric class of correlation functions. A
commonly used two-parameter family due to \cite{matern1960spatial} is
\begin{equation}
    \rho(u;\phi,\kappa) 
= \{2^{\kappa-1} \Gamma(\kappa)\}^{-1} 
(u/\phi)^{\kappa} K_{\kappa}(u/\phi),
\label{eq:matern}
\end{equation}
where $\Gamma(\cdot)$ denotes the gamma function, and $K_\kappa(\cdot)$ denotes the modified Bessel function of the second kind of order $\kappa$. 
This yields the corresponding parametric variogram model
\begin{equation}
\varphi(u;\tau,\sigma,\theta)
=
\tau^2+\sigma^2\{1-\rho(u;\theta)\},
\end{equation}
for $u > 0$,
where $\tau$ and $\sigma$ are nonnegative parameters.
Since the Mat\'ern correlation function is positive definite on the 
Euclidean spatial domain \citep{stein1999interpolation}, the resulting model defines a valid variogram (i.e., a conditionally negative definite function).

 The parameters are estimated by weighted least squares:
\begin{equation}
(\hat{\tau},\hat{\sigma},\hat{\theta})
=
\argmin_{\tau,\sigma,\theta}
G(\tau,\sigma,\theta),
\end{equation}
where $G(\tau,\sigma,\theta)
=
\sum_{k=1}^K
n_k
\{\hat{\varphi}_k-\varphi(u_k;\tau,\sigma,\theta)\}^2$ and $n_k$ denotes the number of distances $u_{ij}$ that belong to the interval
$I_{h,k}$.
The resulting estimator of $\varphi_Y(u)$ is
$\varphi(u;\hat{\tau}, \hat{\sigma}, \hat{\theta})$.

\subsection{Simulation Setup}
We evaluate the predictive performance of Fr\'echet kriging (FK) 
through simulation studies with two types of non-Euclidean 
responses: univariate distributions equipped with the Wasserstein metric 
(Example \ref{exm:unidist}) and SPSD matrices (Example \ref{exm:spd}).
For each setting, we consider sample sizes $n=50,100,200,400$ and 
different nugget effect levels, with 200 Monte Carlo replications for each scenario.  
We compare FK with 
local Fr\'echet regression (LFR) and global Fr\'echet regression (GFR) 
\citep{petersen2019frechet}. For univariate distributions, we additionally 
compare FK with a kriging approach based on Wasserstein principal component 
analysis (WPCA), as used by \citet{qiu2026spatial}.

In each Monte Carlo run, $n$ sampling locations are independently generated 
from the uniform distribution on the square $[-5,5]\times[-5,5] \subset \mathbb{R}^2$. 
Predictive performance is evaluated using the mean squared prediction error 
(MSPE) over 50 independently generated test locations in the same square, which are held fixed across Monte Carlo runs. 
Specifically, for the $q$th Monte Carlo run, let $\widehat{S}_q$ denote the 
predictor and $S_q$ the latent process. The MSPE is defined as
\begin{equation}
    \mathrm{MSPE}_q
    =
    \frac{1}{50}
    \sum_{i=1}^{50}
    d^2\left\{
        \widehat{S}_q(x_i^{\mathrm{test}}),
        S_q(x_i^{\mathrm{test}})
    \right\},
\end{equation}
where $d$ denotes the metric on the corresponding metric space.

In the simulations, we implemented FK using the variogram model based on 
the Mat\'ern correlation function \eqref{eq:matern} with $\kappa=1$, and 
the number of intervals was set to $K=10$. 
LFR and GFR were implemented using the corresponding functions in the 
\texttt{frechet} package \citep{chen2023frechet}. The bandwidths for LFR 
were selected by cross-validation over $[0.5,4]\times[0.5,4]$.
For WPCA, we retained the minimum number of principal components required 
to explain more than 80\% of the total variance and applied ordinary 
kriging separately to each retained principal component score.

\paragraph{Data-Generating Process for Univariate Distributions.}
We define a stochastic process $\{S(x):x\in \mathbb{R}^2\}$ taking values in the 
2-Wasserstein space $\mathcal{P}_2(\mathbb{R})$, where, for each 
$x\in \mathbb{R}^2$, $S(x)$ is the normal distribution $S(x)=N(\mu(x),\sigma^2(x))$. 
Here, $\{\mu(x): x \in \mathbb{R}^2\}$ is a real-valued zero-mean Gaussian process with a constant variance $\text{Var}(\mu(x)) = 1$ and correlation function $\text{Corr}(\mu(x), \mu(y)) = \exp(-\|x-y\|^2/25)$. $\{\sigma(x): x \in \mathbb{R}^2\}$ is a positive-valued stochastic process defined by $\sigma(x) = e^{Z(x)}+1$, where $\{Z(x): x \in \mathbb{R}^2\}$ is a real-valued zero-mean Gaussian process independent of $\mu(\cdot)$, with a constant variance $\text{Var}(Z(x)) = 0.5$ and correlation function $\text{Corr}(Z(x), Z(y)) = \exp(-\|x-y\|^2/25)$. 

We then define the observation process $\{Y(x): x \in \mathbb{R}^2\}$ by 
\begin{equation}
    Y(x) = N(\mu(x)+\varepsilon_1(x), (\sigma(x)+\varepsilon_2(x))^2),
\end{equation}
where for each $k=1, 2$, $\{\varepsilon_k(x): x \in \mathbb{R}^2\}$ is an error process whose marginal distribution is 
the uniform distribution on $[-a,a]$, and the random variables $\varepsilon_k(x), x \in \mathbb{R}^2$ are mutually independent. 
The constant $a \in [0, 1)$ controls the nugget effect level. 
The corresponding perturbation map and verification
of Assumptions~\ref{ass:const_mean} and
\ref{ass:second_stationary}
are provided in Section~\ref{sec:verify_univdist}
of the Supplement. We  consider three levels of the nugget effect: $a=0, 0.4, 0.8$. 

To better reflect real-world settings in which the underlying probability 
distributions are not directly observable, we generate independent samples 
from each distribution. Specifically, for each $Y_i$, we generate 300 
independent observations $\{y_{ij}\}_{j=1}^{300}$ from $Y_i$ and use their 
empirical distribution as the observed response; see also 
\citet{zhou2024wasserstein}.

% Equivalently, for each $x\in \mathbb{R}^2$, the perturbation map is given by
% $
% \mathcal P_x(\omega)
% =
% \Psi^{-1}\left\{
%     \Psi(\omega)
%     +\varepsilon_1(x)
%     +\varepsilon_2(x)\Phi^{-1}(\cdot)
% \right\}
% $
% for
% $
% \omega\in\operatorname{supp}\{S(x)\}
% $,
% where $\Psi:\mathcal{P}_2(\mathbb{R})\to L^2([0,1])$ is the quantile 
% embedding defined by $\Psi(\nu)=F_\nu^{-1}$, $L^2([0,1])$ denotes the 
% Hilbert space of square-integrable functions on $[0,1]$, and $\Phi^{-1}$ denotes the quantile 
% function of the standard normal distribution.
% Under this data-generating process, Assumptions 
% \ref{ass:const_mean} and \ref{ass:second_stationary} are satisfied (see Section \ref{sec:verify_univdist} in the Supplement for details).
 %Thus, all prediction methods are applied to 
% the empirical distributions rather than the underlying distributions 
% $Y_i$.
% This sampling step introduces additional estimation error and
% conditional bias, even when $a=0$, which are not covered by
% the perturbation model \eqref{eq:pertmat_mean}.
% The simulations therefore also assess prediction
% performance in the presence of finite-sample distribution estimation
% error.

\paragraph{Data-Generating Process for SPSD Matrices.}
The power metric with exponent $1/2$ between two SPSD matrices $P, Q \in \text{Sym}_3^{+}$ is given by $d(P,Q)=\|P^{1/2}-Q^{1/2}\|_F$ where $\|\cdot\|_F$ denotes the Frobenius norm.
Here, the matrix power is defined as $P^{1/2} = VD^{1/2} V'$, where $VD V'$ is the usual spectral decomposition of $P \in \text{Sym}_3^{+}$. 
% Although the latent and observed matrices generated below are positive
% definite, we regard them as elements of the SPSD space
% $\mathrm{Sym}_3^{+}$ when constructing the predictor.

We define a $\mathrm{Sym}_3^{+}$-valued stochastic process $\{S(x): x \in \mathbb{R}^2\}$ by $S(x)=U\Lambda(x)U'$ where $U$ is the orthogonal matrix
\[
U=
\begin{pmatrix}
\sqrt{3}/2 & -1/2 & 0 \\
1/2 & \sqrt{3}/2 & 0 \\
0 & 0 & 1
\end{pmatrix},
\]
and $\Lambda(x)=\mathrm{diag}(\lambda_1(x),\lambda_2(x),\lambda_3(x))$ is a diagonal matrix.
For each $k=1,2,3$, we generate the positive-valued stochastic process $\lambda_k(x)=\exp\{Z_k(x)\}+1$ where $\{Z_k(x):x\in \mathbb{R}^2\}$ is a zero-mean Gaussian process with constant variance
$
\mathrm{Var}\{Z_k(x)\}=1
$
and correlation function
$
\mathrm{Corr}\{Z_k(x),Z_k(y)\}=\exp(-\|x-y\|^2/25).
$

We then define the observation process $\{Y(x): x \in \mathbb{R}^2\}$ by
\begin{equation}
    Y(x) = U\{\Lambda(x)^{1/2}+E(x)\}^2U',
\end{equation}
where $E(x) = \mathrm{diag}(\varepsilon_1(x),\varepsilon_2(x),\varepsilon_3(x))$ is a diagonal matrix. 
Here, for each $k$, $\{\varepsilon_k(x): x \in \mathbb{R}^2\}$ is an error process whose marginal distribution is 
the uniform distribution on $[-a,a]$, and the random variables ${\varepsilon_k(x), x \in \mathbb{R}^2}$ are mutually independent. The constant $a \in [0, 1)$ controls the nugget effect level.
The corresponding perturbation map and verification
of Assumptions~\ref{ass:const_mean} and
\ref{ass:second_stationary}
are provided in Section~\ref{sec:verify_SPD}
of the Supplement.
In the simulation, we consider three levels of the nugget effect: $a=0, 0.4, 0.8$.

\subsection{Simulation Results}
\begin{figure}[p]
\single
    \centering
    \includegraphics[width=0.98\linewidth]{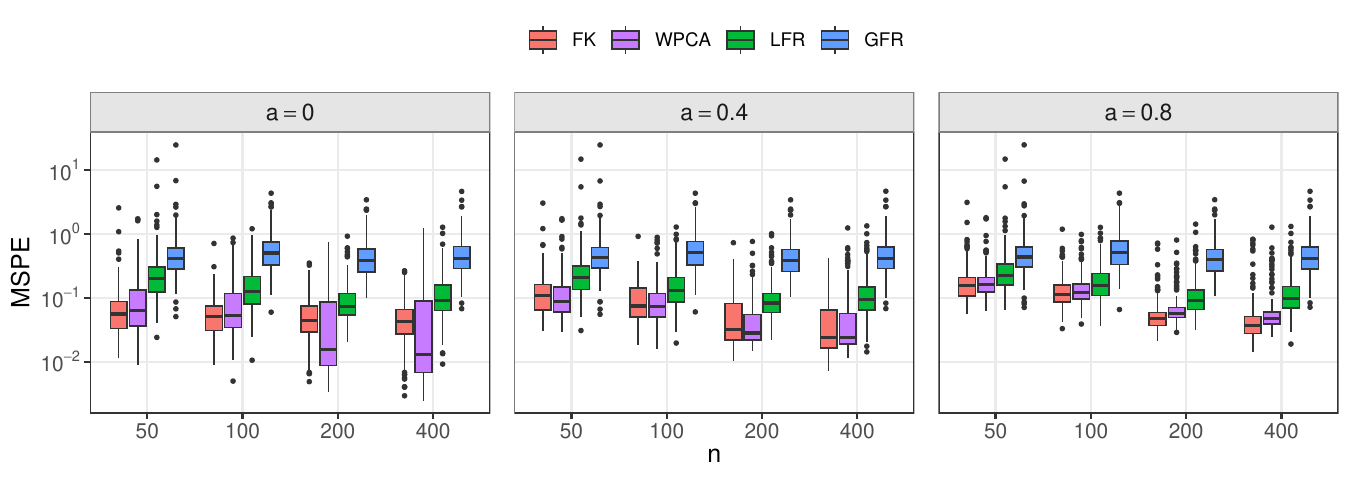}
    \caption{Boxplots of mean squared prediction errors (MSPEs) for univariate distributions,  based on 200 Monte Carlo replications. FK denotes the proposed Fr\'echet kriging, WPCA denotes the method proposed by \cite{qiu2026spatial}, and LFR and GFR denote the local and global Fr\'echet regressions of \cite{petersen2019frechet}.}
    \label{fig:sim_dist_results}
\end{figure}

\begin{figure}[p]
\single
    \centering
    \includegraphics[width=0.98\linewidth]{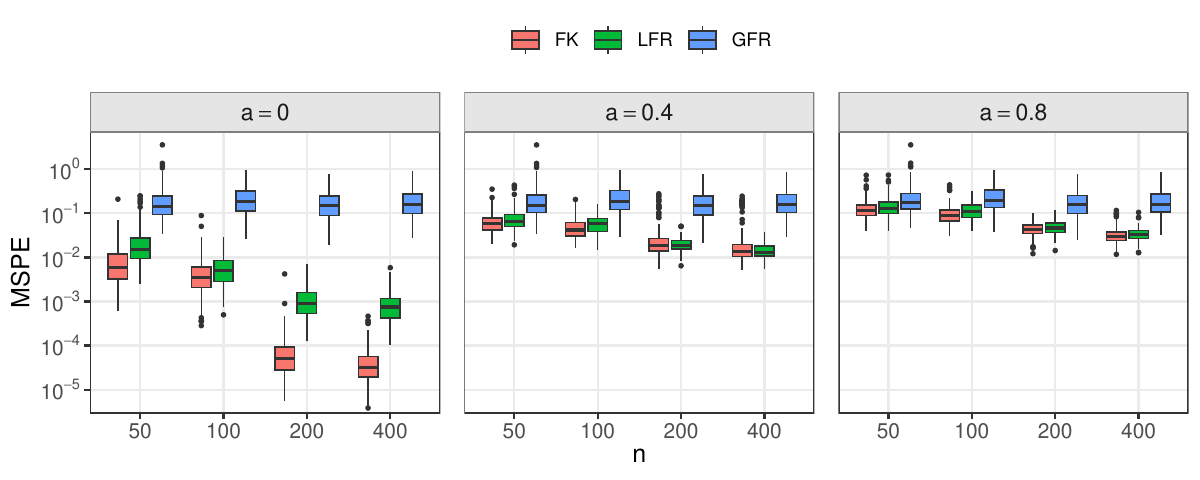}
    \caption{Boxplots of mean squared prediction errors (MSPEs) for SPSD matrices,  based on 200 Monte Carlo replications. FK denotes the proposed Fr\'echet kriging, and LFR and GFR denote the local and global Fr\'echet regressions of \cite{petersen2019frechet}.}
    \label{fig:sim_cov_results}
\end{figure}

Figure~\ref{fig:sim_dist_results} shows the results for univariate distributions,
while Figure~\ref{fig:sim_cov_results} shows those for SPSD matrices.
In most settings, the proposed FK method outperforms LFR and GFR,
particularly when $a=0$.
This suggests that explicitly accounting for spatial dependence
among observations can improve predictive accuracy.
For univariate distributions, WPCA sometimes achieves better
predictive performance than FK.
The improved performance of WPCA suggests that exploiting
the specific structure of the Wasserstein space can be beneficial
for spatial prediction.
In contrast, FK does not rely on such a structure and
can be applied to general metric space-valued responses.

% For univariate distributions, FK generally exhibits favorable predictive 
% performance across the sample sizes and nugget effect levels considered. 
% When $a=0$, the MSPE of FK decreases as the sample size increases and 
% is smaller than those of LFR and GFR across all sample sizes. FK also 
% performs comparably to WPCA. Similar patterns are observed when $a=0.4$ 
% and $0.8$: FK remains competitive with WPCA and generally outperforms LFR 
% and GFR. 

% For SPD matrices, FK exhibits particularly strong performance when there 
% is no nugget effect. When $a=0$, its MSPE decreases substantially as 
% the sample size increases and becomes markedly smaller than those of LFR 
% and GFR for larger sample sizes. As the nugget effect increases, the 
% difference between FK and LFR becomes smaller. Nevertheless, FK slightly 
% outperforms LFR for $a=0.4$ and $0.8$, while both methods consistently 
% outperform GFR.

% Overall, the simulation results demonstrate that FK provides competitive 
% predictive performance for both distribution- and SPD matrix-valued 
% spatial data across a range of sample sizes and nugget effect levels. 
% In particular, FK performs strongly when the nugget effect is absent or small and remains competitive as the nugget effect increases. These results suggest that FK can effectively exploit spatial dependence for prediction across varying levels of the nugget effect.

\section{Real Data Analysis}\label{sec:real_data}
In this section, we illustrate the practical utility of the proposed framework through applications to  multivariate distributions (Example \ref{exm:multdist}) and networks (Example \ref{exm:network}). These  two examples demonstrate that the proposed framework can be applied to data types for which kriging methods have not been extensively studied. An additional real data analysis using univariate distributions is provided
in Section \ref{sec:app_univariate_prob} in the Supplement. 

\subsection{Joint Probability Densities of Daily Minimum and Maximum Temperatures}\label{sec:app_twodim_prob}
We consider daily minimum and maximum temperatures recorded in 2020
at meteorological stations in Iowa and the surrounding region of the
United States. To mitigate potential spatial nonstationarity, we restrict
the analysis to a relatively homogeneous geographic region characterized
by low-relief topography and the absence of major topographic barriers.
The data were obtained from the Global Historical Climatology Network
Daily (GHCNd) dataset
(\url{https://www.ncei.noaa.gov/products/land-based-station/global-historical-climatology-network-daily}).
We selected stations with sufficiently complete records (at least 350 daily observations) and low elevations (elevation $\leq 500$ m).
The resulting 236 stations are displayed in Figure \ref{fig:app_maps_1} in the Supplement.

Using this dataset, we predict the joint probability densities of daily minimum and maximum temperatures. To reduce temporal heterogeneity and potential spatial nonstationarity,
we restrict the analysis to the June--September period.
For each station $i=1,\ldots,236$, let 
$
(Z_{i1}^S,W_{i1}^S),\ldots,(Z_{im_i^S}^S,W_{im_i^S}^S),
$
denote the pairs of daily minimum and maximum temperatures observed during the June--September period, where $m_i^S$ denotes the number of observations at station $i$ during the June--September period,
and the superscript $S$ indicates the summer period. 
Based on these observations, we construct a two-dimensional probability density 
function $f_i$ for the joint distribution of daily minimum and maximum 
temperatures, following the procedure of \citet{zhu2024spherical}. 
Specifically, we first perform histogram smoothing on a $50\times50$ grid 
and then apply local linear kernel smoothing. For the local linear smoothing 
step, we use the \texttt{Lwls2D} function in the \texttt{fdapace} package \citep{fdapace}
with bandwidths
$
({\max_{i,t} Z_{it}^S-\min_{i,t} Z_{it}^S})/5
$
and
$
({\max_{i,t} W_{it}^S-\min_{i,t} W_{it}^S})/{5},
$
respectively. 
We then truncate negative fitted values at zero and normalize the resulting estimate so that its numerical integral over
the grid equals one.
We equip the space of two-dimensional probability density functions with the Fisher--Rao metric and treat each density function $f_i$ as an element of the resulting metric space (Example \ref{exm:multdist}).
The density preprocessing is performed once using all 236 stations, and the resulting density functions are treated as fixed throughout the evaluation.

We first estimated the variogram using the procedure described in Section~\ref{sec:implement}, with $K=10$ and the Mat\'ern correlation function \eqref{eq:matern}. 
Figure~\ref{fig:app_dist_vario_2} shows the resulting empirical variogram together with the fitted model. 
The fitted variogram closely follows the empirical estimates, especially at short distances, which are particularly important for kriging prediction. 
The estimated parameters are $\hat{\tau}^2=0.024$ and
$\hat{\tau}^2+\hat{\sigma}^2=0.079$.

\begin{figure}[t]
\single
  \begin{minipage}[b]{0.49\linewidth}
    \centering
\includegraphics[width=0.9\columnwidth]{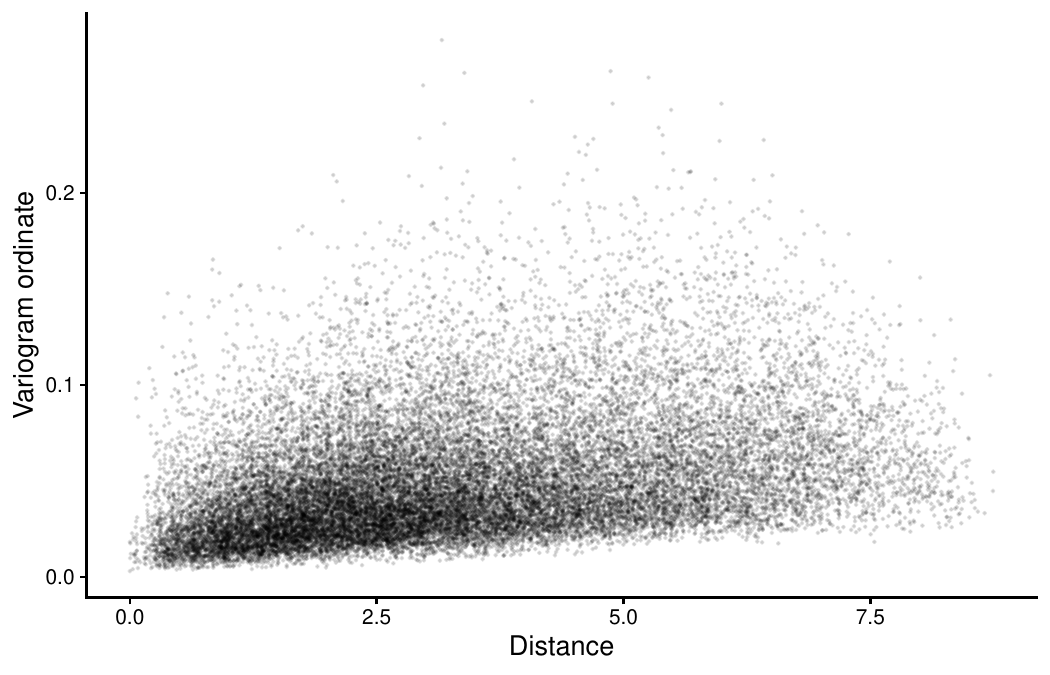}
    \subcaption{Empirical variogram}
  \end{minipage}
  \begin{minipage}[b]{0.49\linewidth}
    \centering
\includegraphics[width=0.9\columnwidth]{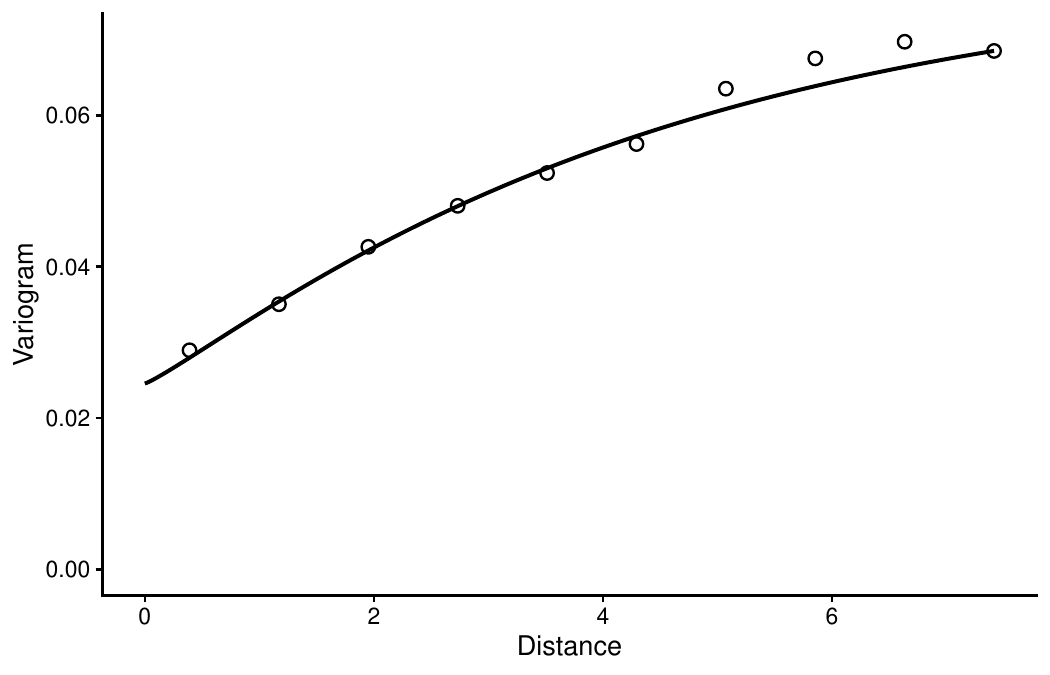}
    \subcaption{Mat\'ern variogram fit}
  \end{minipage}
  \caption{Results of variogram estimation for the joint densities of daily minimum and maximum temperatures in Iowa. Panel (a) shows the empirical variogram based on Euclidean distances between longitude--latitude coordinates (in degrees) and the corresponding variogram ordinates. Panel (b) shows the raw variogram estimates (circles) together with the fitted Mat\'ern variogram model (solid line).}
  \label{fig:app_dist_vario_2}
\end{figure}

We next evaluate the predictive performance of the proposed method
against LFR and GFR using leave-one-out cross-validation.
Under the Fisher--Rao metric, we approximately minimize the
weighted Fr\'echet objective using an iterative algorithm
on the unit sphere of nonnegative square-root densities.
Starting from a normalized weighted extrinsic mean based on
the absolute weights, the algorithm iteratively updates the
estimate using spherical logarithmic and exponential maps,
followed by truncation to the positive orthant and normalization.
The algorithm stops when the update norm falls below $10^{-8}$
or after 200 iterations.
For each cross-validation fold, the variogram
parameters of the proposed method are estimated
using only the training data. For LFR, the bandwidth parameters are selected by cross-validation within each training fold. 
Across the 236 held-out stations, the mean squared Fisher--Rao
prediction errors were 0.0281 for Fr\'echet kriging,
0.0294 for LFR, and 0.0347 for GFR.
Among the methods considered, Fr\'echet kriging achieved the
lowest mean cross-validation error, with reductions of approximately
4.42\% and 19.02\% relative to LFR and GFR, respectively.

% \begin{table}[t]
% \single
%     \centering
%     \begin{tabular}{ccc}
%        FK & LFR&GFR \\ \hline  
%      \textbf{0.0281}(0.0219) &0.0294(0.0225) &0.0347(0.0215) 
%     \end{tabular}
%    \caption{Summary of the leave-one-out cross-validation errors for the
% joint densities of daily minimum and maximum temperatures in Iowa,
% based on squared Fisher--Rao distance over 236 held-out stations.
% Values are mean site-level prediction errors, with standard deviations
% of the site-level errors in parentheses.
% FK denotes the proposed Fr\'echet kriging method, and LFR and GFR denote
% local and global Fr\'echet regression, respectively
% \citep{petersen2019frechet}.}
%     \label{tab:app_dist_2}
% \end{table}

\begin{figure}
\single
    \centering
    \includegraphics[width=0.99\linewidth]{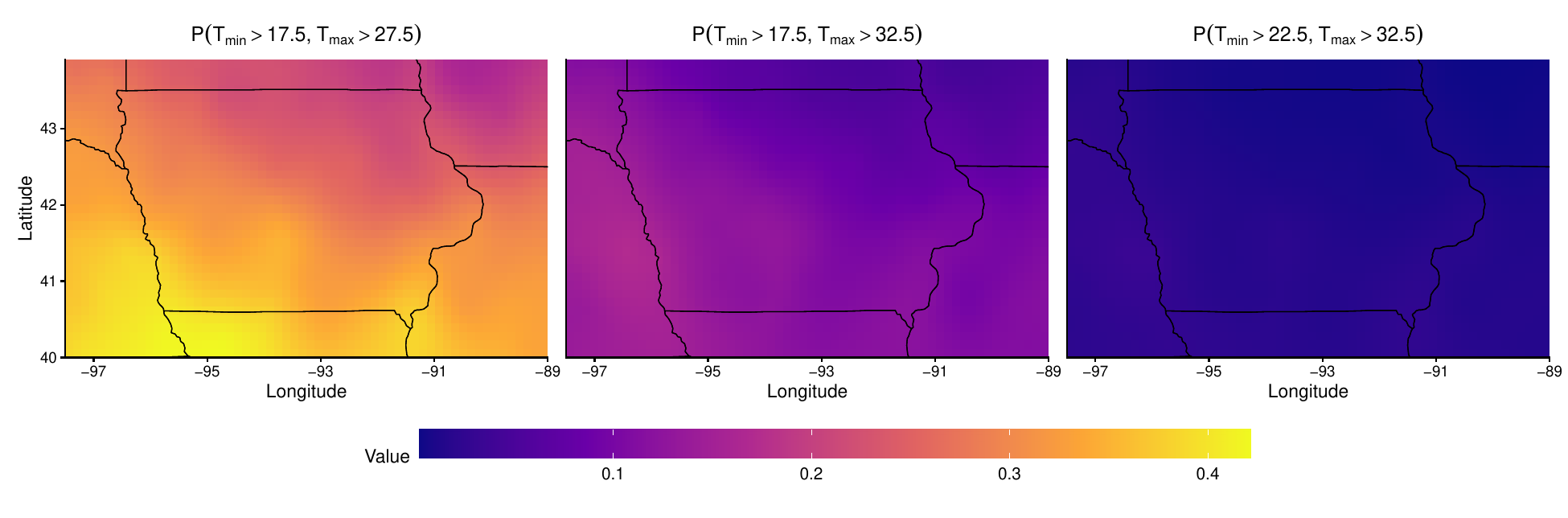}
    \caption{Predicted probability maps by the proposed method, $P(T_{\min}>17.5,\,T_{\max}>27.5)$ (left),
$P(T_{\min}>17.5,\,T_{\max}>32.5)$ (center), and
$P(T_{\min}>22.5,\,T_{\max}>32.5)$ (right).}
    \label{fig:mesh_2}
\end{figure}

Finally, we divide the spatial domain into a $60 \times 30$ grid and predict the joint probability density function at the center of each grid cell. Using the predicted densities, we then calculate several joint exceedance probabilities. Figure~\ref{fig:mesh_2} presents the predicted probabilities
$P(T_{\min}>17.5,\,T_{\max}>27.5)$,
$P(T_{\min}>17.5,\,T_{\max}>32.5)$, and
$P(T_{\min}>22.5,\,T_{\max}>32.5)$.
The resulting probability maps exhibit smooth spatial variation, reflecting the spatial dependence captured by the proposed method. As the thresholds for $T_{\min}$ and $T_{\max}$ increase, the corresponding joint exceedance probabilities decrease, as expected.

\subsection{Inter-Industry Monetary Flow Networks}\label{sec:app_nwk}
We consider a spatial network dataset derived from the 2018 European 
Multi-Regional Input--Output (EMRIO) dataset \citep{huang2023european}. To better satisfy the stationarity assumption, we restrict our analysis to 61 regions in Germany, Belgium, the Netherlands, and Luxembourg. The study region is shown in Figure~\ref{fig:Germany} in the Supplement.

For each region, we construct an undirected $10\times10$ weighted network 
representing monetary flows among 10 major economic sectors. Following the 
NACE Level 1 classification used in the EMRIO dataset, the sectors are 
Primary, Industry, Construction, Distribution, ICT, Financial services, 
Real estate activities, Professional services, Public services, and Other 
services. The original monetary flows are measured 
in millions of US dollars.
Let $X_{ij}$ denote the original monetary flow from sector $i$ to sector 
$j$ within a region. For each pair of distinct sectors $i$ and $j$, we 
define the total monetary flow between them as $X_{ij}+X_{ji}$, thereby 
constructing an undirected network. Since the original monetary flows are 
highly right-skewed, we apply the transformation $\log(1+x)$ to the total 
flows. Specifically, the $(i,j)$th entry of the adjacency matrix is defined as $A_{ij}=\log\{1+X_{ij}+X_{ji}\}$, $i\neq j$, with $A_{ii}=0$. Thus, each region is represented by a symmetric weighted adjacency matrix 
with zero diagonal entries (i.e., no self-loops), whose edge weights quantify 
the log-transformed total monetary flows between economic sectors. We then represent each network 
by its graph Laplacian matrix and equip the space of graph Laplacian 
matrices with the Frobenius metric (Example~\ref{exm:network}).

The empirical and estimated variograms are plotted in Figure \ref{fig:app_dist_vario_3}. 
We estimated the variogram following the procedure described in Section \ref{sec:implement}, setting $K=11$ and using the cardinal sine correlation function \citep{chiles1999geostatistics}, $\rho(u; \phi) = \frac{\sin(u/\phi)}{u/\phi}$
to accommodate the non-monotone behavior observed in the empirical variogram.
As shown in the figure, the fitted variogram captures the overall non-monotone 
pattern of the empirical variogram, including its initial increase followed by 
a subsequent decrease.
The estimated nugget variance and sill are $\hat{\tau}^2=304.6$ and 
$\hat{\tau}^2+\hat{\sigma}^2=546.7$, respectively. 
Thus, the estimated nugget variance accounts for approximately 55.7\% of the 
sill, indicating a substantial nugget effect, while the remaining 44.3\% is 
attributable to spatially structured variation.

\begin{figure}[t]
\single
  \begin{minipage}[b]{0.49\linewidth}
    \centering
\includegraphics[width=0.9\columnwidth]{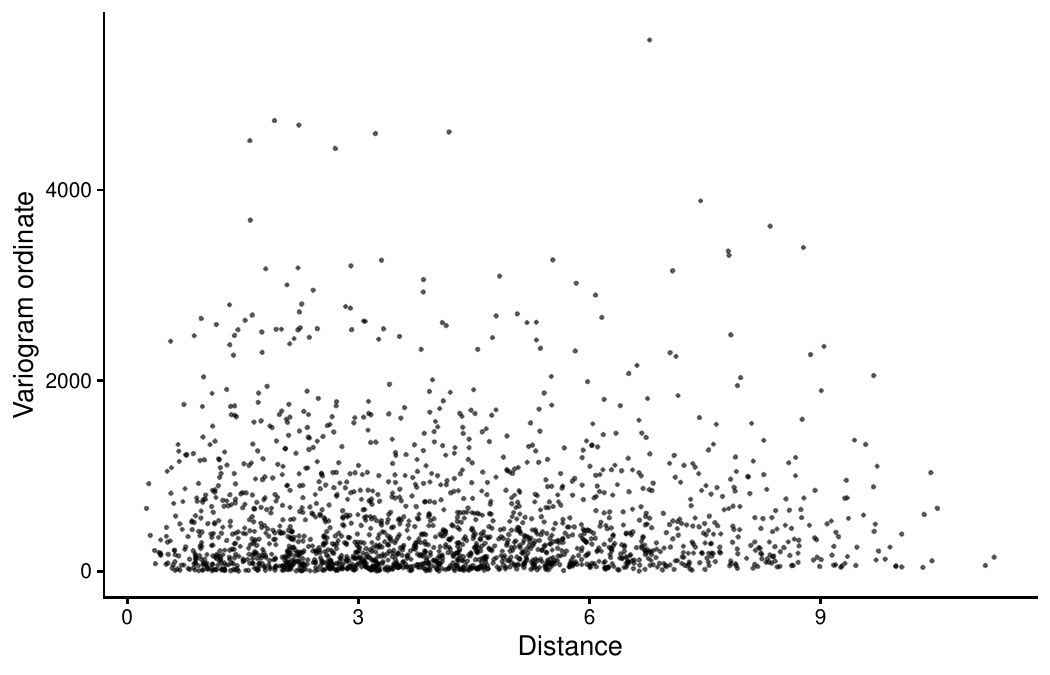}
    \subcaption{Empirical variogram}
  \end{minipage}
  \begin{minipage}[b]{0.49\linewidth}
    \centering
\includegraphics[width=0.9\columnwidth]{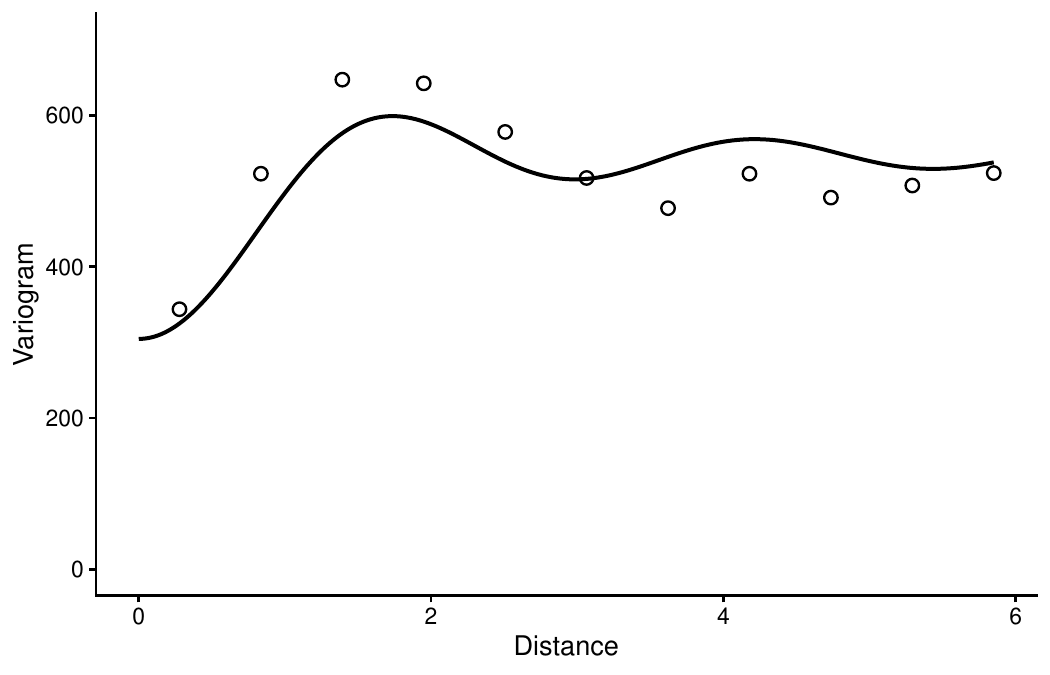}
    \subcaption{Cardinal sine variogram fit}
  \end{minipage}
  \caption{Results of variogram estimation for the monetary flow network data. Panel (a) shows the empirical variogram based on Euclidean distances between longitude--latitude coordinates (in degrees) and the corresponding variogram ordinates. Panel (b) shows the raw variogram estimates (circles) together with the fitted cardinal sine variogram model (solid line).}
  \label{fig:app_dist_vario_3}
\end{figure}

Next, using leave-one-out cross-validation, we compare the prediction performance of the proposed method with that of local and global network regression of \cite{zhou2022network}. 
For each cross-validation fold, the variogram
parameters of the proposed method are estimated
using only the training data.
For the local network regression, the bandwidth parameters are selected by cross-validation within each training fold.
Across the 61 held-out regions, the mean squared Frobenius
prediction errors were 468.7 for Fr\'echet kriging,
527.8 for local network regression, and 551.5 for global
network regression.
Thus, Fr\'echet kriging achieved the lowest mean
cross-validation error, with reductions of approximately
11.2\% and 15.0\% relative to local and global network
regression, respectively.

% \begin{table}[t]
% \single
%     \centering
%     \begin{tabular}{ccc}
%        FK & LNR&GNR \\ \hline  
%      \textbf{468.7}(550.3) & 527.8(587.5)&551.5(649.7) 
%     \end{tabular}
%    \caption{Summary of the leave-one-out cross-validation errors for the
% network data, based on squared Frobenius distance over 61 held-out
% regions. Values are mean site-level prediction errors, with standard
% deviations of the site-level errors in parentheses.
% FK denotes the proposed Fr\'echet kriging method, while LNR and GNR denote
% the local and global network regression methods, respectively
% \citep{zhou2022network}.}
%     \label{tab:app_dist_3}
% \end{table}

We further selected two regions whose cross-validation errors under the 
proposed method are closest to the first and third quartiles of the errors 
across all regions. Figure \ref{fig:one_shot_nwk} shows the heatmaps of the observed networks and those predicted by the proposed method. Although the discrepancy is somewhat larger for Region B, 
the displayed adjacency patterns are visually similar.

\begin{figure}[t]
\single
    \centering
    \includegraphics[width=0.99\linewidth]{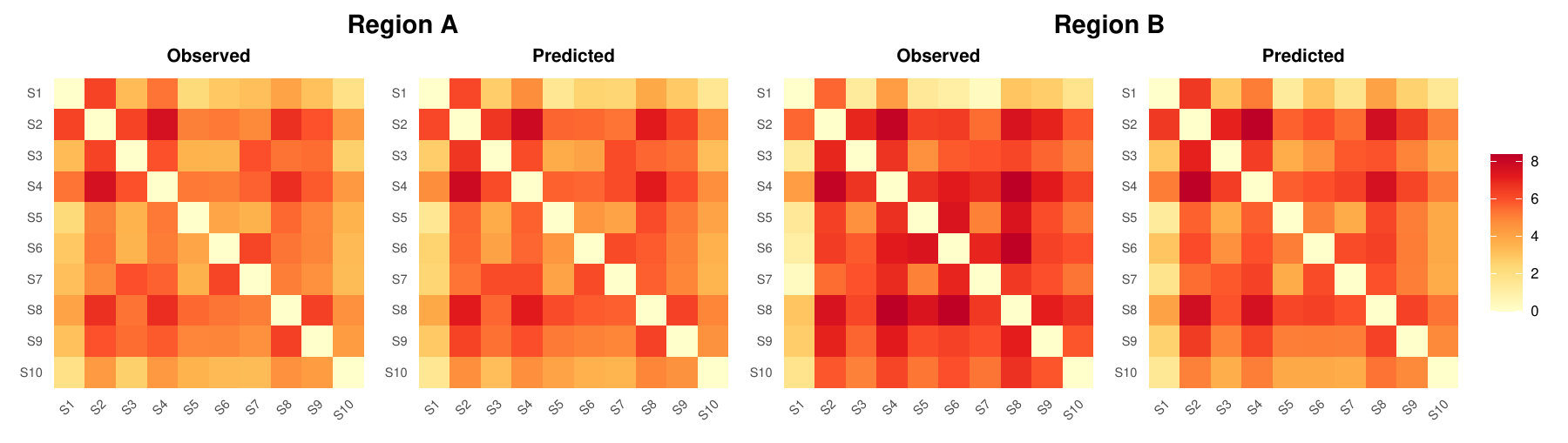}
    \caption{Heatmaps of the observed networks and those predicted by the 
proposed method in two regions. Regions A and B are located at 
(longitude, latitude) $=(5.80, 53.05)$ and $=(4.37, 50.85 )$, respectively. 
The squared Frobenius distances between the observed and predicted graph 
Laplacians are 73.6 and 604.8 for Regions A and B, respectively.}
    \label{fig:one_shot_nwk}
\end{figure}

\section{Concluding Remarks}\label{sec:conc_rmk}

We propose Fr\'echet kriging, a unified framework for
spatial prediction of metric space-valued data. The proposed method  
provides an intrinsic extension of ordinary kriging and relies only on
distances between responses, without requiring an embedding into a
linear space. We establish theoretical properties of the proposed
predictor and demonstrate  its practical utility through simulation
studies and real data applications involving various types of metric
space-valued responses.
The current consistency theory concerns a fixed working variogram with
stable kriging weights and vanishing observation error. Establishing
corresponding guarantees for predictors based on data-dependent
estimated variograms remains an important direction for future work.

Another important direction for future research is to extend
Fr\'echet kriging to accommodate spatially varying mean structures,
analogous to universal kriging \citep{cressie2015statistics}. Such an
extension requires an appropriate formulation of the trend and residual
spatial dependence for metric space-valued responses, where subtraction
and other linear operations are generally unavailable.
A further promising direction is an extension to spatio-temporal data
\citep{cressie2011statistics}. By characterizing dependence jointly
over spatial and temporal lags, Fr\'echet kriging could be extended to
predict metric space-valued processes evolving over both space and time.

\single
\bibliographystyle{apalike}
\bibliography{ref}
\double

 \newpage
 \begin{center}
 {\large\bf Supplement}
 \end{center}
\setcounter{section}{0}
 \renewcommand{\thesection}{S.\arabic{section}}
 \setcounter{lem}{0}
\renewcommand{\thelem}{S.\arabic{lem}}

\renewcommand{\theequation}{S\arabic{equation}}
\setcounter{equation}{0}

\setcounter{figure}{0}
\renewcommand{\thefigure}{S.\arabic{figure}}

\section{Details on Remark \ref{rmk:manifold}}
\label{sec:kriging_manifold}
We first provide an alternative representation of the Fr\'echet 
kriging predictor under Assumptions~\ref{ass:isometric} and 
\ref{ass:convex}.  Define the weighted average of the embedded responses 
in the Hilbert space $\mathcal H$ and its metric projection onto $\Psi(\Omega)$ by
\begin{equation}
    \check{Z}(x_0)
=
\sum_{i=1}^n \lambda_i(x_0) \Psi(Y_i)
\quad \text{and} \quad 
\tilde{Z}(x_0)
=
\argmin_{z \in \Psi(\Omega)}
\|z-\check{Z}(x_0)\|_{\mathcal H}.
\end{equation}
By Assumption~\ref{ass:convex}, $\tilde{Z}(x_0)$ exists and is unique.

\begin{lem}\label{lem:alt_rep}
Under Assumptions~\ref{ass:isometric} and \ref{ass:convex}, the 
Fr\'echet kriging predictor \eqref{eq:fk} admits the representation
\begin{equation}
    \tilde{S}(x_0)
    =
    \Psi^{-1}(\tilde{Z}(x_0)).
\end{equation}
\end{lem}
The proof is deferred to Section \ref{sec:proof_lem1}. 
Lemma \ref{lem:alt_rep} shows that, whenever the metric space admits an isometric embedding onto a closed convex subset of a Hilbert space, Fr\'echet kriging is equivalent to kriging in the ambient Hilbert space, followed by projection onto the embedded closed convex subset and application of the inverse embedding.

We next use Lemma~\ref{lem:alt_rep} to establish a connection between 
Fr\'echet kriging and the kriging approach for manifold-valued responses 
proposed by \citet{pigoli2016kriging}. Let $\Omega$ be a Riemannian 
manifold. For a reference point $\mu\in\Omega$, let
$\mathrm{Log}_{\mu}:\Omega\to T_{\mu}$ and
$\mathrm{Exp}_{\mu}:T_{\mu}\to\Omega$ denote the logarithmic and 
exponential maps, respectively, where $T_{\mu}$ is the tangent space 
at $\mu$. A projection-corrected version of the kriging predictor of
\citet{pigoli2016kriging} is defined by
\begin{equation}
    \tilde{S}_M(x_0)
    =
    \mathrm{Exp}_{\mu}(\tilde{Z}_M(x_0)),
    \quad
    \tilde{Z}_M(x_0)
    =
    \argmin_{z\in\mathrm{Log}_{\mu}(\Omega)}
    \left\|
    \sum_{i=1}^n
    \lambda_i^M(x_0)\mathrm{Log}_{\mu}(Y_i)-z
    \right\|_{T_{\mu}}.
\end{equation}
Here, $\|\cdot\|_{T_{\mu}}$ denotes the norm on the tangent space 
$T_{\mu}$, and the weights $(\lambda_i^M(x_0))_{i=1}^n$ are obtained 
as the solution to
\begin{equation}
    \min_{\lambda\in\mathbb{R}^n}
    E\left[
    \left\|
    \sum_{i=1}^n\lambda_i\mathrm{Log}_{\mu}(Y_i)
    -
    \mathrm{Log}_{\mu}(S(x_0))
    \right\|_{T_{\mu}}^2
    \right]
    \quad \text{subject to} \quad
    \sum_{i=1}^n\lambda_i=1.
\end{equation}
Note that the original predictor proposed by
\citet{pigoli2016kriging} does not include a metric projection onto
$\mathrm{Log}_{\mu}(\Omega)$. Since the exponential map is the inverse
of the logarithmic map on $\mathrm{Log}_{\mu}(\Omega)$, it is natural
to project the tangent-space predictor onto
$\mathrm{Log}_{\mu}(\Omega)$ before mapping it back to $\Omega$.

Suppose that the logarithmic map preserves distances, that is,
\begin{equation}
    d(x,y)
    =
    \|\mathrm{Log}_{\mu}(x)-\mathrm{Log}_{\mu}(y)\|_{T_{\mu}}
\end{equation}
for all $x,y\in\Omega$, and that $\mathrm{Log}_{\mu}(\Omega)$ is closed 
and convex in $T_{\mu}$. 
This is the case, for example, when $\Omega$ is  the space of SPSD matrices equipped with the square-root metric \citep[Section~2]{pigoli2016kriging}, or the Wasserstein space of univariate distributions with an atomless reference measure $\mu$ \citep{panaretos2020invitation}.
Then, we may take 
$\Psi=\mathrm{Log}_{\mu}$ in Assumption~\ref{ass:isometric}, with 
$\Psi^{-1}=\mathrm{Exp}_{\mu}$. Therefore, by 
Lemma~\ref{lem:alt_rep}, the Fr\'echet kriging predictor can be 
represented as
\begin{equation}
    \tilde{S}(x_0)
    =
    \mathrm{Exp}_{\mu}(\tilde{Z}(x_0)),
    \quad
    \tilde{Z}(x_0)
    =
    \argmin_{z\in\mathrm{Log}_{\mu}(\Omega)}
    \left\|
    \sum_{i=1}^n
    \lambda_i(x_0)\mathrm{Log}_{\mu}(Y_i)-z
    \right\|_{T_{\mu}}.
\end{equation}
Moreover, since the logarithmic map preserves distances, the criterion 
used to determine the weights $\lambda^M(x_0)$ coincides with that used 
to determine the Fr\'echet kriging weights $\lambda(x_0)$. 
Hence, $\lambda^M(x_0)=\lambda(x_0)$, and thus
$\tilde{S}_M(x_0)=\tilde{S}(x_0)$. Therefore, under these conditions, the Fr\'echet kriging predictor coincides with
the projection-corrected kriging predictor defined above.

% \subsection{Details on Remark~\ref{rmk:dist}}
% \label{sec:rmk_dist}

% Let $(\Omega,d)$ be the $2$-Wasserstein space of univariate distributions
% and let $\mathcal{H}=L^2[0,1]$ be the Hilbert space of square-integrable
% functions on $[0,1]$. Then Assumptions~\ref{ass:isometric} and
% \ref{ass:convex} hold with the map
% $
% \Psi(\mu)=Q_\mu.
% $
% Hence, by Lemma~\ref{lem:alt_rep}, the quantile function
% $\tilde{Z}(x_0)$ of the Fr\'echet kriging predictor $\tilde{S}(x_0)$ is
% given by
% \begin{equation}
% \tilde{Z}(x_0)
% =
% \argmin_{z\in\Psi(\Omega)}
% \|z-\check{Z}(x_0)\|_{\mathcal H},
% \qquad
% \check{Z}(x_0)
% =
% \sum_{i=1}^n
% \lambda_i(x_0)Q_{Y_i}.
% \end{equation}

% As Equation~(9) in \citet{gossard2025kriging} indicates, their kriging
% weights $(\overline{\lambda}_i(x_0))_{i=1}^n$ coincide with our kriging
% weights $(\lambda_i(x_0))_{i=1}^n$. Therefore, both predictors first use the
% same weights to construct a weighted average of the quantile functions
% as an initial predictor. They differ in the subsequent modification of
% this initial predictor: our approach uses the metric projection onto the
% closed convex set $\Psi(\Omega)$, whereas \citet{gossard2025kriging}
% use a sorting procedure.

\section{Verification of the Assumptions for the Simulation Settings}
\subsection{Univariate Distributions}\label{sec:verify_univdist}
Under the data-generating process for univariate distributions
described in the main text, the corresponding perturbation map
is given by
\[
\mathcal P_x(\omega)
=
\Psi^{-1}\left\{
    \Psi(\omega)
    +\varepsilon_1(x)
    +\varepsilon_2(x)\Phi^{-1}(\cdot)
\right\},
\]
for each $x\in\mathbb{R}^2$ and
$\omega\in\operatorname{supp}\{S(x)\}$,
where $\Psi:\mathcal{P}_2(\mathbb{R})\to L^2([0,1])$
is the quantile embedding defined by
$\Psi(\nu)=F_\nu^{-1}$,
$L^2([0,1])$ is the Hilbert space of square-integrable
functions on $[0,1]$, and $\Phi^{-1}$ denotes
the quantile function of the standard normal distribution.

We now verify that Assumptions~\ref{ass:const_mean}
and \ref{ass:second_stationary} are satisfied.
The quantile 
function of $S(x)=N\{\mu(x),\sigma^2(x)\}$ is given by
\[
    F_{S(x)}^{-1}(u)
    =
    \mu(x)+\sigma(x)\Phi^{-1}(u),
    \qquad u\in(0,1).
\]
Since the 2-Wasserstein metric is the $L^2$ distance between quantile 
functions, the Fr\'echet mean $\mu_S(x)$ of $S(x)$ has the quantile function
\begin{equation}
    F_{\mu_S(x)}^{-1}(u)
    =
    E[\mu(x)]
    +
    E[\sigma(x)]\Phi^{-1}(u)
    =
    E[\sigma(x)]\Phi^{-1}(u),
\end{equation}
where we used $E[\mu(x)]=0$. Since the marginal distribution of $Z(x)$ does not depend on $x$, the
distribution of $\sigma(x)=\exp\{Z(x)\}+1$ also does not depend on $x$. Therefore, 
$\mu_S(x)$ is constant in $x$. The Fr\'echet variance of $S(x)$ is
\begin{align}
    \sigma_S^2(x)
    &=
    E\left[
    d_W^2\{S(x),\mu_S(x)\}
    \right] 
    =
    E\int_0^1
    \left[
    \mu(x)
    +
    \{\sigma(x)-E[\sigma(x)]\}\Phi^{-1}(u)
    \right]^2du \\
    &=
    \operatorname{Var}(\mu(x))
    +
    \operatorname{Var}(\sigma(x)),
\end{align}
because $\int_0^1 \Phi^{-1}(u)\,du=0$, $\int_0^1 \{\Phi^{-1}(u)\}^2du=1$.
Since both $\operatorname{Var}(\mu(x))$ and 
$\operatorname{Var}(\sigma(x))$ are constant in $x$, 
$\sigma_S^2(x)$ is also constant.
The variogram of $S$ is given by
\begin{align}
    V_S(x,y)
    =
    \frac{1}{2}
    E\left[
    d_W^2\{S(x),S(y)\}
    \right] 
    =
    \frac{1}{2}
    E\left[
    \{\mu(x)-\mu(y)\}^2
    +
    \{\sigma(x)-\sigma(y)\}^2
    \right],
\end{align}
for $x \neq y$.
Since $(\mu(x),\mu(y))$ is jointly Gaussian, its distribution is completely determined by its mean vector and covariance matrix. In our setting, the mean vector is constant, while the covariance matrix depends on the correlation
$
\mathrm{Corr}\{\mu(x),\mu(y)\}=\exp(-\|x-y\|^2/25),
$
which is a function only of the distance $\|x-y\|$. Therefore, the joint distribution of $(\mu(x),\mu(y))$ depends only on $\|x-y\|$.
Similarly, 
the joint distribution of $(Z(x),Z(y))$ depends only on $\|x-y\|$. Since 
$(\sigma(x),\sigma(y))$ is a measurable function of $(Z(x),Z(y))$, its 
joint distribution also depends only on $\|x-y\|$. It follows 
that both
$
    E[\{\mu(x)-\mu(y)\}^2]
$
and
$
    E[\{\sigma(x)-\sigma(y)\}^2]
$
depend only on $\|x-y\|$. Hence, $V_S(x,y)$ is a function only of the 
separation distance $\|x-y\|$.

For the observation process $Y$, its quantile function is
\[
    F_{Y(x)}^{-1}(u)
    =
    \mu(x)+\varepsilon_1(x)
    +
    \{\sigma(x)+\varepsilon_2(x)\}\Phi^{-1}(u).
\]
Since $E\{\varepsilon_k(x)\}=0$ for $k=1,2$, the Fr\'echet mean of 
$Y(x)$ has quantile function
\begin{align}
    F_{\mu_Y(x)}^{-1}(u)
    &=
    E\{\mu(x)+\varepsilon_1(x)\}
    +
    E\{\sigma(x)+\varepsilon_2(x)\}\Phi^{-1}(u) 
    =
    E\{\sigma(x)\}\Phi^{-1}(u)
    =
    F_{\mu_S(x)}^{-1}(u).
\end{align}
Thus, $\mu_Y(x)=\mu_S(x)$ and is constant in $x$.
Moreover, since
$
    \operatorname{Var}\{\varepsilon_k(x)\}
    =
    \frac{a^2}{3}, k=1, 2, 
$
and the error terms are independent of the latent process, we have
\begin{align}
    \sigma_Y^2(x)
    =
    \operatorname{Var}\{\mu(x)+\varepsilon_1(x)\}
    +
    \operatorname{Var}\{\sigma(x)+\varepsilon_2(x)\} 
    =
    \sigma_S^2(x)+\frac{2a^2}{3}.
\end{align}
Hence, the Fr\'echet variance of $Y(x)$ is also constant in $x$.
Finally, for $x \neq y$,
\begin{align}
    V_Y(x,y)
    &=
    \frac{1}{2}
    E\left[
    \{\mu(x)+\varepsilon_1(x)-\mu(y)-\varepsilon_1(y)\}^2
    \right] 
    +
    \frac{1}{2}
    E\left[
    \{\sigma(x)+\varepsilon_2(x)-\sigma(y)-\varepsilon_2(y)\}^2
    \right] \\
    &=
    V_S(x,y)
    +
    \frac{2a^2}{3}.
\end{align}
Therefore, both $S$ and $Y$ have constant Fr\'echet means and 
variances, and their variograms depend only on the distance $\|x-y\|$. 
Hence, Assumptions~\ref{ass:const_mean} and 
\ref{ass:second_stationary} are satisfied.

\subsection{SPSD Matrices}\label{sec:verify_SPD}
Under the data-generating process for SPSD Matrices
described in the main text, the corresponding perturbation map
is given by
$$
\mathcal P_x(\omega)
=
\Psi^{-1}\{\Psi(\omega)+UE(x)U^\top\},
$$
for each $x \in \mathbb{R}^2$ and 
$
\omega\in\operatorname{supp}\{S(x)\},
$
where $\Psi: \text{Sym}_3^{+} \to \text{Sym}_3^{+}$ is defined by $\Psi(P) = P^{1/2}$.

We now verify that Assumptions~\ref{ass:const_mean}
and \ref{ass:second_stationary} are satisfied.
The Fr\'echet mean and variance of $S(x)$ are given by 
\begin{align}
    \mu_{S}(x) 
    &= U\operatorname{diag}\left(
\{E[\lambda_1(x)^{1/2}]\}^2,
\{E[\lambda_2(x)^{1/2}]\}^2,
\{E[\lambda_3(x)^{1/2}]\}^2
\right)U', \\
    \sigma^2_{S}(x) &= \sum_{k=1}^3 \text{Var}(\lambda_k(x)^{1/2}).
\end{align}
Since the distribution of $\lambda_k(x)$ does not depend on $x$ for each $k$, both $\mu_S(x)$ and $\sigma_S^2(x)$ are constant. 
The variogram of $S$ is given by 
\begin{equation}
V_S(x, y) = 
    \frac{1}{2}E[d\{S(x), S(y)\}^2]
    =\frac{1}{2}
    \sum_{k=1}^3 E[\{\lambda_k(x)^{1/2} -\lambda_k(y)^{1/2} \}^2],
\end{equation}
for $x \neq y$.
Since $(Z_k(x),Z_k(y))$ is jointly Gaussian, its distribution is completely determined by its mean vector and covariance matrix. In our setting, the mean vector is constant, while the covariance matrix depends on the correlation
$
\mathrm{Corr}\{Z_k(x),Z_k(y)\}=\exp(-\|x-y\|^2/25),
$
which is a function only of the separation distance $\|x-y\|$. Therefore, the joint distribution of $(Z_k(x),Z_k(y))$ depends only on $\|x-y\|$. Since
$
\{
\lambda_k(x)^{1/2}
-
\lambda_k(y)^{1/2}
\}^2
$
is a measurable function of $(Z_k(x),Z_k(y))$, it follows that
$
E[\{
\lambda_k(x)^{1/2}
-
\lambda_k(y)^{1/2}\}^2]
$
also depends only on $\|x-y\|$. 

Moreover, we have 
\begin{align}
\mu_Y(x)
&=
U\operatorname{diag}\Bigl(\!
    \{E[\lambda_1(x)^{1/2}\!+\varepsilon_1(x)]\}^2,
    \{E[\lambda_2(x)^{1/2}\!+\varepsilon_2(x)]\}^2,
    %\nonumber\\
%&\hspace{5.5cm}
    \{E[\lambda_3(x)^{1/2}\!+\varepsilon_3(x)]\}^2
\Bigr)U'
\nonumber\\
&=
\mu_S(x),
\\[0.5em]
\sigma_Y^2(x)
&=
\sum_{k=1}^3
\operatorname{Var}\{\lambda_k(x)^{1/2}+\varepsilon_k(x)\}
%\nonumber\\
=
\sigma_S^2(x)+a^2,
\\[0.5em]
V_Y(x,y)
&=
\frac{1}{2}
\sum_{k=1}^3
E\Bigl[
\bigl\{
\lambda_k(x)^{1/2}+\varepsilon_k(x)
-\lambda_k(y)^{1/2}-\varepsilon_k(y)
\bigr\}^2
\Bigr]
%\nonumber\\
=
V_S(x,y)+a^2.
\end{align}
for $x \neq y$.
This means Assumptions \ref{ass:const_mean} and \ref{ass:second_stationary} are satisfied.

\section{Proofs}

\subsection{Proof of Lemma \ref{lem:alt_rep}}\label{sec:proof_lem1}
\begin{proof}
For $x_0\in\mathbb{R}^p$, define the objective function
\begin{equation}
    M_n(x_0,\omega)
    =
    \sum_{i=1}^n
    \lambda_i(x_0)d^2(Y_i,\omega),
    \qquad \omega\in\Omega.
\end{equation}
Under Assumption~\ref{ass:isometric}, and using
$\sum_{i=1}^n\lambda_i(x_0)=1$, the objective function can be written as
\begin{align}
    M_n(x_0,\omega)
    &=
    \sum_{i=1}^n
    \lambda_i(x_0)
    \|\Psi(Y_i)-\Psi(\omega)\|_{\mathcal H}^2 \\
    &=
    \sum_{i=1}^n
    \lambda_i(x_0)
    \|\Psi(Y_i)-\check Z(x_0)\|_{\mathcal H}^2
    +
    \|\check Z(x_0)-\Psi(\omega)\|_{\mathcal H}^2.
\end{align}
Since the first term does not depend on $\omega$, it follows that
\begin{equation}
    \tilde S(x_0)
    =
    \argmin_{\omega\in\Omega}
    \|\check Z(x_0)-\Psi(\omega)\|_{\mathcal H}^2 
    =
    \Psi^{-1}\left\{
    \argmin_{z\in\Psi(\Omega)}
    \|\check Z(x_0)-z\|_{\mathcal H}^2
    \right\} 
    =
    \Psi^{-1}\{\tilde Z(x_0)\}.
\end{equation}
This completes the proof.
\end{proof}

\subsection{Proof of Proposition~\ref{prp:opt}}
\begin{proof}
For each $\lambda\in\Lambda_n$, by an argument similar to that in 
Lemma~\ref{lem:alt_rep}, the weighted Fr\'echet mean predictor 
$\tilde{S}(x_0;\lambda)$ in \eqref{eq:general_predictor} can be represented as
\begin{equation}
    \tilde{S}(x_0;\lambda)
    =
    \Psi^{-1}\{\tilde{Z}(x_0;\lambda)\},
\end{equation}
where
\begin{equation}
    \check{Z}(x_0;\lambda)
    =
    \sum_{i=1}^n \lambda_i \Psi(Y_i),
    \qquad
    \tilde{Z}(x_0;\lambda)
    =
    \argmin_{z\in\Psi(\Omega)}
    \|z-\check{Z}(x_0;\lambda)\|_{\mathcal H}.
\end{equation}
Together with the isometric property of $\Psi$, this implies that
\begin{equation}
    R(x_0;\lambda)
    =
    E\left[
    \|\tilde{Z}(x_0;\lambda)-\Psi\{S(x_0)\}\|_{\mathcal H}^2
    \right].
    \label{eq:risk_representation}
\end{equation}
Since the metric projection onto the nonempty closed and convex set 
$\Psi(\Omega)$ in the Hilbert space $\mathcal H$ is 1-Lipschitz
\citep[Section~4.3]{bauschke2017convex}, and since 
$\Psi(S(x_0))\in\Psi(\Omega)$, we have
\begin{equation}
    \|\tilde{Z}(x_0;\lambda)-\Psi(S(x_0))\|_{\mathcal H}
    \le
    \|\check{Z}(x_0;\lambda)-\Psi(S(x_0))\|_{\mathcal H}.
    \label{eq:projection_inequality}
\end{equation}
Combining \eqref{eq:risk_representation} and
\eqref{eq:projection_inequality} and taking expectations yields the
conclusion of part~(i).

For part (ii), we first show that
\begin{equation}
E\!\left[\|\Psi(Y_i)-\Psi(Y(x_0))\|_{\mathcal H}^2\right]
=
E\!\left[\|\Psi(Y_i)-\Psi(S(x_0))\|_{\mathcal H}^2\right]
+\tau^2.
\label{eq:vario_equality}
\end{equation}
To see this, write
$
\varepsilon_0
=
\Psi(Y(x_0))-\Psi(S(x_0)).
$
Then
\begin{align}
E\!\left[\|\Psi(Y_i)-\Psi(Y(x_0))\|_{\mathcal H}^2\right]
={}&
E\!\left[\|\Psi(Y_i)-\Psi(S(x_0))\|_{\mathcal H}^2\right]
+E\!\left[\|\varepsilon_0\|_{\mathcal H}^2\right] \\
&-2E\!\left[
\left\langle
\Psi(Y_i)-\Psi(S(x_0)),\varepsilon_0
\right\rangle_{\mathcal H}
\right].
\label{eq:decompose_hilbert}
\end{align}
Since the perturbation map $\mathcal{P}_{x_0}$ is independent of $S(x_0)$, by the isometric property of $\Psi$ and
\eqref{eq:pertmat_var},
\begin{equation}
    E\left[\|\varepsilon_0\|_{\mathcal H}^2\right]
    =
    E\left[
        d^2\{\mathcal P_{x_0}(S(x_0)),S(x_0)\}
    \right]
    =
    E\left[E\left[
        d^2\{\mathcal P_{x_0}(S(x_0)),S(x_0)\}
     \mid S(x_0)\right]\right]
    =
    \tau^2.
    \label{eq:var_tau2}
\end{equation}
Moreover, let
$
m_0 = E[\Psi\{Y(x_0)\}\mid S(x_0)],
$
where the conditional expectation is understood in the Bochner sense in the Hilbert space $\mathcal H$.
Since $\Psi\{Y(x_0)\}\in\Psi(\Omega)$ and
$\Psi(\Omega)$ is closed and convex, we have
$m_0\in\Psi(\Omega)$ \citep[Chapter II]{diestel1977vector}. Moreover, the conditional expectation $m_0$
is the unique minimizer of
$
E\left[
\|\Psi\{Y(x_0)\}-y\|_{\mathcal H}^2
\mid S(x_0)
\right]
$
over $y\in\mathcal H$, and hence also over $y\in\Psi(\Omega)$.
By the isometric property of $\Psi$, the independence of $\mathcal{P}_{x_0}$ and $S(x_0)$, and the condition
\eqref{eq:pertmat_mean},
\[
\begin{aligned}
m_0
&=
\argmin_{y\in\Psi(\Omega)}
E\left[
\|\Psi\{Y(x_0)\}-y\|_{\mathcal H}^2
\mid S(x_0)
\right]\\
&=
\Psi\left[
\argmin_{\omega\in\Omega}
E\left[
d^2\{\mathcal{P}_{x_0}(S(x_0)),\omega\}
\mid S(x_0)
\right]
\right]\\
&=
\Psi\{S(x_0)\}.
\end{aligned}
\]
Therefore,
\[
E[\varepsilon_0\mid S(x_0)]
=
E[\Psi\{Y(x_0)\}\mid S(x_0)]-\Psi\{S(x_0)\}
=0.
\]
 Since $\mathcal{P}_{x_0}$ is
independent of the latent process $S$ and of $\mathcal{P}_{x_i}$, the preceding
centering property implies
\[
E[\varepsilon_0\mid S, \mathcal{P}_{x_i}]
=
E[\varepsilon_0\mid S]
=
E[\varepsilon_0\mid S(x_0)]
=
0.
\]
Moreover, by the tower property of conditional expectation,
\begin{align}
E\left[
\left\langle
\Psi(Y_i)-\Psi\{S(x_0)\},
\varepsilon_0
\right\rangle_{\mathcal H}
\right]
=
E\left[
\left\langle
\Psi(Y_i)-\Psi\{S(x_0)\},
E[\varepsilon_0\mid S, \mathcal{P}_{x_i}]
\right\rangle_{\mathcal H}
\right]
=0.
\label{eq:inner_0}
\end{align}
Combining \eqref{eq:decompose_hilbert}, \eqref{eq:var_tau2} and  \eqref{eq:inner_0}, we obtain \eqref{eq:vario_equality}.

Next, let $\Gamma_H$ be the $n\times n$ matrix whose $(i,j)$th element is
$
    E[
    \|\Psi(Y(x_i))-\Psi(Y(x_j))\|_{\mathcal H}^2
    ]/2,
$
 $\gamma_H(x_0)$ be the $n$-dimensional vector whose $i$th element is
$
    E\left[
    \|\Psi(Y(x_0))-\Psi(Y(x_i))\|_{\mathcal H}^2
    \right]/2,
$
and $c_H(x_0)$ be the $n$-dimensional vector whose $i$th element is
$
    E\left[
    \|\Psi(S(x_0))-\Psi(Y(x_i))\|_{\mathcal H}^2
    \right]/2.
$
The above argument implies that $c_H(x_0) = \gamma_H(x_0) - (\tau^2/2) \mathbf{1}_n$. Using $\sum_{i=1}^n \lambda_i = 1$, a direct expansion gives
\begin{equation}
    E[\|\check{Z}(x_0;\lambda)-\Psi(S(x_0))\|_{\mathcal H}^2]
    =
    -\lambda^\top \Gamma_H \lambda + 2\lambda^{\top}c_H(x_0)
    =
    -\lambda^\top \Gamma_H \lambda + 2\lambda^{\top}\gamma_H(x_0)-\tau^2.
\end{equation}
 The minimizer of the right hand side of the above expression  is given by 
\begin{equation}
    \lambda_H(x_0)
    =
    \begin{pmatrix}
        I_n & \mathbf{0}_n
    \end{pmatrix}
    \begin{pmatrix}
        \Gamma_H & \mathbf{1}_n \\
        \mathbf{1}_n^\top & 0
    \end{pmatrix}^{-1}
    \begin{pmatrix}
        \gamma_H(x_0) \\
        1
    \end{pmatrix}.
\end{equation}
By the isometric property of $\Psi$, $\Gamma_H$ and $\gamma_H(x_0)$
coincide with $\Gamma$ and $\gamma(x_0)$ in
\eqref{eq:weights_vario}, respectively. Hence,
$\lambda_H(x_0)=\lambda(x_0)$, which proves part~(ii).
\end{proof}

\subsection{Proof of Theorem \ref{thm:miss}}

For the proof of Theorem \ref{thm:miss}, the following lemma plays a key role. Its proof is deferred to
Section~\ref{sec:average}.

\begin{lem}\label{lem:average}
Under Assumption~\ref{ass:miss}(ii), (iv), and (v), we have
\begin{equation}
    T_n g
    :=
    \sum_{i=1}^n
    \lambda_{\psi,i}(x_0)g(x_i)
    \to g(x_0)
    \qquad \text{as } n\to\infty,
    \label{eq:scalar_reproduction}
\end{equation}
for every $g\in C(A)$, where $C(A)$ denotes the space of
continuous real-valued functions on $A$.
\end{lem}

We now prove Theorem \ref{thm:miss}.

\begin{proof}
 We prove the conclusion in two steps.
 
\noindent
\textbf{(Step 1)} Define 
\begin{align*}
\Delta_n &:=\sup_{\omega\in\Omega}
\left|\sum_{i=1}^n
\lambda_{\psi,i}(x_0)
d^2\{S(x_i),\omega\}-d^2(S(x_0),\omega)\right|.
\end{align*}
In this step, we will show 
\begin{equation}\label{lim:FR1}
\lim_{n \to \infty}\Delta_n= 0\quad \text{almost surely}.
\end{equation}
For $\omega \in \Omega$, define $g_\omega(x)=d^2(S(x),\omega)$. By the triangle inequality, we have $|g_\omega(x)-g_\omega(y)| \leq 2D_\Omega d\{S(x),S(y)\}$, where $D_\Omega$ denotes the diameter of $(\Omega,d)$, which is finite by the boundedness of $(\Omega,d)$. Since $A$ is compact and the sample path $x \mapsto S(x)$ is almost surely continuous, the family $\mathcal{G}_S=\{g_\omega: \omega\in \Omega\}$ is almost surely uniformly bounded and equicontinuous. By the Arzel\`a-Ascoli theorem, we then have that $\mathcal G_S$ is relatively compact in $C(A)$ almost surely. 

Let
\[
T_ng=\sum_{i=1}^n\lambda_{\psi,i}(x_0)g(x_i),
\qquad g\in C(A).
\]
Applying Lemma \ref{lem:average}, $T_ng\to g(x_0)$ for every $g\in C(A)$.
Furthermore, the operator norm of $T_n$ is uniformly bounded with respect to $n$, that is, $\sup_n\|T_n\|=\sup_n\sum_{i=1}^n|\lambda_{\psi,i}(x_0)|<\infty$.
Consequently, the uniform boundedness of the operators $T_n$
and their pointwise convergence imply uniform convergence on
every relatively compact subset of $C(A)$. Therefore, 
\begin{align*}
\Delta_n 
&=
\sup_{g\in\mathcal G_S}
\left|T_ng-g(x_0)\right|
\to 0\quad \text{almost surely}
\end{align*}
and hence we obtain \eqref{lim:FR1}. 

%\textcolor{red}{
%\noindent
%\textbf{(Step 2)} In this step, we will show 
%\begin{equation}\label{lim:FR2}
%\mathbb E\left[
%d_{\mathrm{FR}}^2(\hat{S}_{\psi}^o(x_0),S(x_0))
%\right] \to  0\quad \text{as $n \to \infty$}.
%\end{equation}
%Define the nugget-free predictor based on the working variogram by
%\begin{equation}
%    \hat{S}_{\psi}^o(x_0)
%    =
%    \arg\min_{f \in \mathcal{P}(D)}
%    \sum_{i=1}^n
%    \lambda_{\psi,i}(x_0)
%    d_{\mathrm{FR}}^2\{S(x_i),f\}.
%\end{equation}
%Then we have
%\begin{align*}
%d_{\mathrm{FR}}^2(\hat{S}_{\psi}^o(x_0), S(x_0))
%&\leq \sum_{i=1}^n
%\lambda_{\psi,i}(x_0)
%d_{\mathrm{FR}}^2\{\hat{S}_{\psi}^o(x_0),S(x_i)\} + \Delta_n\\
%&\leq  \sum_{i=1}^n
%\lambda_{\psi,i}(x_0)
%d_{\mathrm{FR}}^2\{S(x_0),S(x_i)\} + \Delta_n
%\leq 2\Delta_n. 
%\end{align*}
%Then, \eqref{lim:FR1} yields $\lim_{n \to \infty}d_{\mathrm{FR}}^2(\hat{S}_{\psi}^o(x_0), S(x_0)) = 0$ almost surely. Since $d_{\mathrm{FR}}^2(\hat{S}_{\psi}^o(x_0),S(x_0))\le\pi^2$,
%the dominated convergence theorem yields \eqref{lim:FR2}.
%}

\noindent
\textbf{(Step 2)} In this step, we complete the proof. Define 
\[
\Delta'_n = \sup_{\omega \in \Omega}\left|\sum_{i=1}^n
\lambda_{\psi,i}(x_0)
d^2(S(x_i),\omega)-\sum_{i=1}^n
\lambda_{\psi,i}(x_0)
d^2(Y_i,\omega)\right|. 
\]
Observe that 
\begin{align}
E\left[\Delta'_n\right]
&\leq 2D_\Omega\sum_{i=1}^n
|\lambda_{\psi,i}(x_0)|E\{d(Y_i,S(x_i))\}
\leq 2D_\Omega ME\{d^2(Y_i,S(x_i))\}^{1/2} \nonumber \\
&=2D_\Omega M\tau_n \to 0\quad \text{as $n \to \infty$}, \label{lim:FR3}
\end{align}
where $M = \sup_{n \ge1}\sum_{i=1}^n |\lambda_{\psi, i}(x_0)| < \infty$. Then, \eqref{lim:FR1} and \eqref{lim:FR3} yield
\begin{align}
d_{}^2(\hat{S}_{\psi}(x_0), S(x_0))
&\leq \sum_{i=1}^n
\lambda_{\psi,i}(x_0)
d_{}^2\{Y_i,\hat{S}_{\psi}(x_0)\} + \Delta'_n + \Delta_n \nonumber \\
&\leq \sum_{i=1}^n
\lambda_{\psi,i}(x_0)
d_{}^2\{Y_i,S(x_0)\} + \Delta'_n + \Delta_n \nonumber \\
&\leq 2(\Delta'_n + \Delta_n) \stackrel{p}{\to} 0 \quad \text{as $n \to \infty$}. \label{lim:FR4}
\end{align}
For the second inequality, we used the optimality of $\hat{S}_{\psi}(x_0)$. Then, for any $\varepsilon>0$, \eqref{lim:FR4} yields
\begin{align*}
E\left[d^2\{\hat{S}_{\psi}(x_0), S(x_0)\}\right]
&\leq \varepsilon + D_\Omega^2 P\left(d^2\{\hat{S}_{\psi}(x_0), S(x_0)\}>\varepsilon\right) \to \varepsilon \quad \text{as $n \to \infty$}. 
\end{align*}
Therefore, for every $\varepsilon>0$,
$
\limsup_{n\to\infty}
E\left[
d^2\{\hat{S}_{\psi}(x_0),S(x_0)\}
\right]
\leq \varepsilon.
$
Since the expectation is nonnegative and $\varepsilon>0$ is arbitrary, letting $\varepsilon\downarrow0$ yields
$
\lim_{n\to\infty}
E\left[
d^2\{\hat{S}_{\psi}(x_0),S(x_0)\}
\right]
=0,
$
which completes the proof.
\end{proof}

\subsection{Proof of Remark \ref{rmk:without_bounded}}\label{sec:without_bounded}
\begin{proof}[Proof of Remark \ref{rmk:without_bounded}]
Define the nugget-free predictor based on the working variogram by
\begin{equation}
    \hat{S}_{\psi}^o(x_0)
    =
    \arg\min_{\omega\in\Omega}
    \sum_{i=1}^n
    \lambda_{\psi,i}(x_0)
    d^2\{S(x_i),\omega\}.
\end{equation}
We first show that
$
    E[d ^2\{\hat{S}_{\psi}^o(x_0), S(x_0)\}] \to 0.
$
Fix a sample path of $S$ that is continuous on $A$. 

Set
$
    U(x)=\Psi\{S(x)\}\in\mathcal H.
$
Since $\Psi$ is an isometry and the sample path of $S$ is continuous,
$U:A\to\mathcal H$ is continuous. Hence, since $A$ is compact,
$U(A)$ is compact. Define
\begin{equation}
    z_n
    =
    \sum_{i=1}^n
    \lambda_{\psi,i}(x_0)U(x_i).
\end{equation}
Fix $\varepsilon>0$. Since $U(A)$ is compact, there exists a
finite-dimensional subspace $\mathcal H_\varepsilon\subset\mathcal H$
with orthogonal projection
$P_\varepsilon:\mathcal H\to\mathcal H_\varepsilon$
such that
\begin{equation}
    \sup_{x\in A}
    \|U(x)-P_\varepsilon U(x)\|_{\mathcal H}
    <\varepsilon.
    \label{eq:approx_U}
\end{equation}
Since $\mathcal H_\varepsilon$ is finite-dimensional,
$P_\varepsilon U(x)$ has finitely many scalar coordinate functions,
each of which is continuous on $A$.
Applying Lemma~\ref{lem:average} to each coordinate yields
\begin{equation}
    \|P_\varepsilon z_n-P_\varepsilon U(x_0)\|_{\mathcal H}
    \to0.
    \label{eq:conv_P_e}
\end{equation}
Moreover, by \eqref{eq:approx_U}, the $\ell_1$ stability of the weights,
and the linearity of $P_\varepsilon$, we have
\begin{align}
    \|(I-P_\varepsilon)z_n
      -(I-P_\varepsilon)U(x_0)\|_{\mathcal H}
    &\le
    \sum_{i=1}^n
    |\lambda_{\psi,i}(x_0)|
    \|(I-P_\varepsilon)U(x_i)\|_{\mathcal H}
    +
    \|(I-P_\varepsilon)U(x_0)\|_{\mathcal H}
    \\
    &\le
    (M+1)\varepsilon,
\end{align}
where $M = \sup_{n \ge1}\sum_{i=1}^n |\lambda_{\psi, i}(x_0)| < \infty$. Combining this with \eqref{eq:conv_P_e}, we obtain
\begin{equation}
    \limsup_{n\to\infty}
    \|z_n-U(x_0)\|_{\mathcal H}
    \le
    (M+1)\varepsilon.
\end{equation}
Since $\varepsilon>0$ is arbitrary, it follows that $z_n\to U(x_0)$ as $n \to \infty$. %\label{eq:zn_conv}

By an argument identical to that used in Lemma~\ref{lem:alt_rep}, and the fact that the metric projection onto a nonempty closed convex subset
of a Hilbert space is $1$-Lipschitz, we obtain
\begin{equation}
    d\{\hat S_{\psi}^o(x_0),S(x_0)\}
    \le
    \|z_n-U(x_0)\|_{\mathcal H}
    \to0.
\end{equation}
The above argument holds for every continuous sample path, and such sample paths occur with probability one. Hence, we have $d\{\hat S_{\psi}^o(x_0),S(x_0)\} \to 0$ a.s.
Moreover, since the kriging weights sum to one and $\Psi$ is an isometry,
\begin{align}
    d\{\hat S_{\psi}^o(x_0),S(x_0)\}
    &\le
    \|z_n-U(x_0)\|_{\mathcal H} 
    =
    \left\|
    \sum_{i=1}^n
    \lambda_{\psi,i}(x_0)
    \{U(x_i)-U(x_0)\}
    \right\|_{\mathcal H} \\
    &\le
    M\sup_{x\in A}d\{S(x),S(x_0)\}.
\end{align}
Hence,
\begin{equation}
    d^2\{\hat S_{\psi}^o(x_0),S(x_0)\}
    \le
    M^2\sup_{x\in A}d^2\{S(x),S(x_0)\}.
\end{equation}
The right-hand side is integrable by the moment condition in
Remark~\ref{rmk:without_bounded}. 
Therefore, the dominated convergence theorem yields
$
    E[
    d^2\{\hat S_{\psi}^o(x_0),S(x_0)\}
    ]
    \to0.
$

Next, we show that $E[d^2\{\hat{S}_{\psi}^o(x_0), \hat{S}_{\psi}(x_0)\}] \to 0$.
By an argument identical to that used in Lemma~\ref{lem:alt_rep},
and the $1$-Lipschitz property of the metric projection,
\begin{align}
    d\{\hat{S}_{\psi}^o(x_0),\hat{S}_{\psi}(x_0)\}
    &\le
    \left\|
    \sum_{i=1}^n
    \lambda_{\psi,i}(x_0)
    \left[
    \Psi\{S(x_i)\}-\Psi(Y_i)
    \right]
    \right\|_{\mathcal H} \\
    &\le
    \sum_{i=1}^n
    |\lambda_{\psi,i}(x_0)|
    d\{S(x_i),Y_i\},
\end{align}
where the second inequality follows from the fact that $\Psi$ is an isometry. 
Therefore, noting that $E\{d^2(S(x_i),Y_i)\}=\tau_n^2$,
Minkowski's inequality yields
\begin{align}
    \left[
    E\left\{
    d^2\{\hat{S}_{\psi}^o(x_0),\hat{S}_{\psi}(x_0)\}
    \right\}
    \right]^{1/2}
    &\le
    \sum_{i=1}^n
    |\lambda_{\psi,i}(x_0)|
    \left[
    E\{d^2(S(x_i),Y_i)\}
    \right]^{1/2} \le
    M\tau_n.
\end{align}
Hence, $E\left[
    d^2\{\hat{S}_{\psi}^o(x_0),\hat{S}_{\psi}(x_0)\}
    \right]
    \le
    M^2\tau_n^2
    \to0$. Finally, by Minkowski's inequality,
\begin{align}
    \left[
    E\!\left\{
    d^2\{\hat{S}_{\psi}(x_0),S(x_0)\}\!
    \right\}
    \right]^{1/2}
    &\!\!\le \!
    \left[
    E\!\left\{
    d^2\{\hat{S}_{\psi}(x_0),\hat{S}_{\psi}^o(x_0)\}\!
    \right\}
    \right]^{1/2}\!\!
    +\!
    \left[
    E\!\left\{
    d^2\{\hat{S}_{\psi}^o(x_0),S(x_0)\}\!
    \right\}
    \right]^{1/2}
    \!\!\to0.
\end{align}
Thus, we have \eqref{eq:converge_without_bounded}, which completes the proof.
\end{proof}

\subsection{Proof of Remark \ref{rmk:unif}}\label{sec:rmk_unif}

For the proof of Remark \ref{rmk:unif}, we use the following lemma, which is a uniform version of
Lemma~\ref{lem:average}. Its proof is deferred to Section~\ref{sec:average_2}.

\begin{lem}\label{lem:average_2}
For $x_0 \in A$ and $g\in C(A)$, define
\begin{equation}
    T_{n,x_0}g
    :=
    \sum_{i=1}^n
    \lambda_{\psi,i}(x_0)g(x_i).
\end{equation}
Under the assumptions of Theorem~\ref{thm:miss} and the additional
conditions stated in Remark~\ref{rmk:unif}, we have
\begin{equation}
    \sup_{x_0\in A}
    |T_{n,x_0}g-g(x_0)|
    \to0
    \qquad\text{as } n\to\infty
\end{equation}
for every $g\in C(A)$.
\end{lem}

We now prove Remark \ref{rmk:unif}.

\begin{proof}[Proof of Remark~\ref{rmk:unif}]
We prove the conclusion in two steps.

\noindent
\textbf{(Step 1)}
For $x_0\in A$, define
\[
\Delta_n(x_0)
=
\sup_{\omega\in\Omega}
\left|
\sum_{i=1}^n
\lambda_{\psi,i}(x_0)
d^2\{S(x_i),\omega\}
-
d^2\{S(x_0),\omega\}
\right|.
\]
In this step, we will show that
\begin{equation}
E\left[
\sup_{x_0\in A}
\Delta_n(x_0)
\right]
\to0.
\label{eq:uniform_delta_latent}
\end{equation}
Fix a sample path of $S$ that is continuous on $A$. For each
$\omega\in\Omega$, define
$
g_\omega(x)
=
d^2\{S(x),\omega\}
$
and let
$
\mathcal G_S
=
\{g_\omega:\omega\in\Omega\}.
$
Since $(\Omega,d)$ is bounded, we have
$
\sup_{\omega\in\Omega}
\|g_\omega\|_\infty
\le
D_\Omega^2,
$
where $D_\Omega$ denotes the diameter of $(\Omega,d)$, which is finite by the boundedness of $(\Omega,d)$.
Moreover, by the triangle inequality,
$
|g_\omega(x)-g_\omega(y)|
\le
2D_\Omega d\{S(x),S(y)\},
$
for all $x, y \in A$.
Since $A$ is compact and the sample path $x\mapsto S(x)$ is
continuous, the family $\mathcal G_S$ is uniformly bounded and
equicontinuous. Hence, by the Arzel\`a--Ascoli theorem,
$\mathcal G_S$ is relatively compact in $C(A)$.

Define the linear operator
$
\mathcal R_n:C(A)\to\ell^\infty(A)
$
by
$
(\mathcal R_ng)(x_0)
=
T_{n,x_0}g-g(x_0),
x_0\in A.
$
Here, $\ell^\infty(A)$ denotes the space of bounded real-valued functions on $A$, equipped with the supremum norm $\|f\|_\infty=\sup_{x\in A}|f(x)|$, and 
$T_{n, x_0}$ is the operator defined in Lemma \ref{lem:average_2}.
Applying Lemma~\ref{lem:average_2}, we have
\[
\|\mathcal R_ng\|_\infty
=
\sup_{x_0\in A}
|T_{n,x_0}g-g(x_0)|
\to0
\]
for every $g\in C(A)$.

We now show that this convergence is uniform over $\mathcal G_S$.
Let
\[
M
=
\sup_{n\geq1}
\sup_{x_0\in A}
\sum_{i=1}^n
|\lambda_{\psi,i}(x_0)|
<
\infty.
\]
The uniform $\ell_1$ stability of the weights implies that
\[
\|\mathcal R_n\|
\le
\sup_{x_0\in A}
\left\{
\sum_{i=1}^n
|\lambda_{\psi,i}(x_0)|
+1
\right\}
\le
M+1.
\]
Fix $\varepsilon>0$. Since $\mathcal G_S$ is relatively compact,
there exist $g_1,\ldots,g_N\in C(A)$ such that, for every
$g\in\mathcal G_S$,
$
\min_{1\leq j\leq N}
\|g-g_j\|_\infty
<
\varepsilon/\{2(M+1)\}.
$
Consequently,
$
\sup_{g\in\mathcal G_S}
\|\mathcal R_ng\|_\infty
\le
\varepsilon/2
+
\max_{1\leq j\leq N}
\|\mathcal R_ng_j\|_\infty.
$
Since $N$ is finite and
$\|\mathcal R_ng_j\|_\infty\to0$ for every $j$, it follows that
$
\sup_{g\in\mathcal G_S}
\|\mathcal R_ng\|_\infty
\to0.
$
Therefore,
\[
\sup_{x_0\in A}
\Delta_n(x_0)
=
\sup_{\omega\in\Omega}
\|\mathcal R_ng_\omega\|_\infty
\to0
\qquad\text{almost surely}.
\]
In addition, we have
$
0
\le
\sup_{x_0\in A}
\Delta_n(x_0)
\le
(M+1)D_\Omega^2.
$
Thus, the dominated convergence theorem yields
\eqref{eq:uniform_delta_latent}.

\noindent
\textbf{(Step 2)}
In this step, we complete the proof. For $x_0\in A$, define
\[
\Delta_n'(x_0)
=
\sup_{\omega\in\Omega}
\left|
\sum_{i=1}^n
\lambda_{\psi,i}(x_0)
d^2\{S(x_i),\omega\}
-
\sum_{i=1}^n
\lambda_{\psi,i}(x_0)
d^2(Y_i,\omega)
\right|.
\]
For any $a,b,\omega\in\Omega$,
\[
\begin{aligned}
|d^2(a,\omega)-d^2(b,\omega)|
&=
|d(a,\omega)-d(b,\omega)|
\{d(a,\omega)+d(b,\omega)\} 
\le
2D_\Omega d(a,b).
\end{aligned}
\]
It follows that
\[
\begin{aligned}
E\{\Delta_n'(x_0)\}
&\le
2D_\Omega
\sum_{i=1}^n
|\lambda_{\psi,i}(x_0)|
E[d\{S(x_i),Y_i\}] \\
&\le
2D_\Omega
\sum_{i=1}^n
|\lambda_{\psi,i}(x_0)|
\left[
E\{d^2(S(x_i),Y_i)\}
\right]^{1/2} \\
&\le
2D_\Omega M\tau_n.
\end{aligned}
\]
Therefore,
\begin{equation}
\sup_{x_0\in A}
E\{\Delta_n'(x_0)\}
\le
2D_\Omega M\tau_n
\to0.
\label{eq:uniform_delta_perturbation}
\end{equation}

For each $x_0\in A$, the optimality of
$\hat S_\psi(x_0)$ gives
\[
\begin{aligned}
d^2\{\hat S_\psi(x_0),S(x_0)\}
&\le
\sum_{i=1}^n
\lambda_{\psi,i}(x_0)
d^2\{Y_i,\hat S_\psi(x_0)\}
+
\Delta_n(x_0)
+
\Delta_n'(x_0) \\
&\le
\sum_{i=1}^n
\lambda_{\psi,i}(x_0)
d^2\{Y_i,S(x_0)\}
+
\Delta_n(x_0)
+
\Delta_n'(x_0) \\
&\le
2\{\Delta_n(x_0)+\Delta_n'(x_0)\}.
\end{aligned}
\]
Taking expectations and then the supremum over $x_0\in A$, we obtain
\[
\begin{aligned}
\sup_{x_0\in A}
E\left[
d^2\{\hat S_\psi(x_0),S(x_0)\}
\right]
&\le
2\sup_{x_0\in A}
E\{\Delta_n(x_0)\}
+
2\sup_{x_0\in A}
E\{\Delta_n'(x_0)\} \\
&\le
2E\left[
\sup_{x_0\in A}
\Delta_n(x_0)
\right]
+
2\sup_{x_0\in A}
E\{\Delta_n'(x_0)\}.
\end{aligned}
\]
The right-hand side converges to zero by
\eqref{eq:uniform_delta_latent} and
\eqref{eq:uniform_delta_perturbation}. Therefore,
\[
\sup_{x_0\in A}
E\left[
d^2\{\hat S_\psi(x_0),S(x_0)\}
\right]
\to0,
\]
which proves the result.
\end{proof}

\subsection{Proof of Lemma \ref{lem:average}}\label{sec:average}
\begin{proof}
We prove the conclusion in several steps. 

\noindent
\textbf{(Step1)} In this step, we will show 
\begin{align}
    0
    \leq
    V\{\lambda_\psi(x_0);x_0\} 
    \to0\ \text{as $n \to \infty$}.
    \label{eq:V_to_zero}
\end{align}
For $\lambda\in\Lambda_n$, define
\begin{equation}
    V(\lambda;x_0)
    =
    2\sum_{i=1}^n
    \lambda_i W(x_i,x_0)
    -
    \sum_{i=1}^n\sum_{j=1}^n
    \lambda_i\lambda_j W(x_i,x_j).
\end{equation}
The spectral representation in
Assumption~\ref{ass:miss}(iv) implies that $W$ is conditionally
negative definite. Indeed, for any finite collection of points
$z_1,\ldots,z_m\in\mathbb R^p$ and coefficients
$a_1,\ldots,a_m\in\mathbb R$
satisfying $\sum_{i=1}^m a_i=0$,
\begin{align}
    -\sum_{i=1}^m\sum_{j=1}^m
    a_i a_j W(z_i,z_j)
    &=
    \int_{\mathbb R^p}
    \left|
        \sum_{i=1}^m
        a_i e^{\mathrm{i}\xi^\top z_i}
    \right|^2
    \nu(d\xi)
    \geq0.
\end{align}
This yields the convexity of $V(\lambda;x_0)$ over $\Lambda_n$.
To see this, let $t\in[0,1]$,
$\lambda=(\lambda_1,\ldots,\lambda_n)$ and
$\mu=(\mu_1,\ldots,\mu_n)$ belong to $\Lambda_n$, and set
$\delta=\lambda-\mu$ and
$W_n=\{W(x_i,x_j)\}_{i,j=1}^n$.
Then $\sum_{i=1}^n\delta_i=0$, and hence
$\delta^\top W_n\delta\leq0$. Therefore,
\begin{align}
    &V\{t\lambda+(1-t)\mu;x_0\}
    -
    tV(\lambda;x_0)
    -
    (1-t)V(\mu;x_0) \nonumber
    =
    t(1-t)\delta^\top W_n\delta
    \leq0.
\end{align}
Moreover, conditional negative definiteness implies
$
    V(\lambda;x_0)\geq0,
$
for $\lambda\in\Lambda_n$
by applying the above inequality to the points
$x_1,\ldots,x_n,x_0$ with coefficients
$\lambda_1,\ldots,\lambda_n,-1$.
The kriging equations are the first-order optimality conditions for
minimizing $V(\lambda;x_0)$ over $\Lambda_n$, and therefore
$\lambda_\psi(x_0)$ is a minimizer of $V(\lambda;x_0)$ over
$\Lambda_n$.

Let $i_n$ be an index of a nearest sampling location to $x_0$, that is, $i_n\in
    \argmin_{1\leq i\leq n}
    \|x_i-x_0\|$ 
and let $e_{i_n}\in\Lambda_n$ denote the $i_n$th standard basis
vector in $\mathbb{R}^n$. By the optimality of $\lambda_\psi(x_0)$,
\begin{align*}
    0
    \leq
    V\{\lambda_\psi(x_0);x_0\} 
    \leq
    V(e_{i_n};x_0) 
    =
    2W(x_{i_n},x_0)
    =
    2\psi(\|x_{i_n}-x_0\|)
    \to0
\end{align*}
and hence we obtain (\ref{eq:V_to_zero}). Here, the last convergence follows from
Assumption~\ref{ass:miss}(ii) and the continuity of $\psi$ at zero.
We verify the latter property as follows.
By Assumption~\ref{ass:miss}(iv), for any $h\in\mathbb R^p$,
$
    \psi(\|h\|)
    =
    \int_{\mathbb R^p}
    \left[
        1-\cos(\xi^\top h)
    \right]
    \nu(d\xi).
$
For each fixed $\xi\in\mathbb R^p$,
$
    1-\cos(\xi^\top h)\to0
$
as $h \to 0$.
Moreover, if $\|h\|\leq1$, then
\begin{align}
    0
    \leq
    1-\cos(\xi^\top h)
    \leq
    \min\left\{
        2,\frac{1}{2}|\xi^\top h|^2
    \right\} 
    \leq
    \min\left\{
        2,\frac{1}{2}\|\xi\|^2
    \right\} 
    \leq
    2(1\wedge\|\xi\|^2).
\end{align}
The dominating function
$2(1\wedge\|\xi\|^2)$ is integrable with respect to $\nu$ by
Assumption~\ref{ass:miss}(iv). Therefore, the dominated convergence
theorem yields
$
    \psi(\|h\|)
    \to0, 
$
as $h \to 0$.
Since $\psi(0)=0$, $\psi$ is continuous at zero. 

\vspace{5pt}
\noindent
\textbf{(Step2)} In this step, we will show 
\begin{equation}
    Q_n(\xi)\to0
    \qquad
    \text{for every }
    \xi\in\operatorname{supp}(\nu).
    \label{eq:Q_pointwise}
\end{equation}
Define
\begin{equation}
    Q_n(\xi)
    =
    e^{\mathrm{i}\xi^\top x_0}
    -
    \sum_{i=1}^n
    \lambda_{\psi,i}(x_0)
    e^{\mathrm{i}\xi^\top x_i}.
\end{equation}
Since $\lambda_\psi(x_0)\in\Lambda_n$, the spectral representation
in Assumption~\ref{ass:miss}(iv) yields
\begin{equation}
    V\{\lambda_\psi(x_0);x_0\}
    =
    \int_{\mathbb R^p}
    |Q_n(\xi)|^2
    \nu(d\xi).
    \label{eq:spectral_V}
\end{equation}
Combining this with \eqref{eq:V_to_zero}, we obtain
\begin{equation}
    \int_{\mathbb R^p}
    |Q_n(\xi)|^2
    \nu(d\xi)
    \to0.
    \label{eq:Q_L2_zero}
\end{equation}
Additionally, let $M
    =
    \sup_{n\geq1}
    \sum_{i=1}^n
    |\lambda_{\psi,i}(x_0)|
    <\infty$ which is finite by Assumption~\ref{ass:miss}(v).
Since $A$ is compact and contains $x_0,x_1,\ldots,x_n$, we may define $C_A=\sup_{x\in A}\|x\|<\infty$.
Using the inequality $|e^{\mathrm{i}a}-e^{\mathrm{i}b}|\leq|a-b|$,
we have
\begin{align}
    |Q_n(\xi)-Q_n(\eta)|
    &\leq
    |(\xi-\eta)^\top x_0|
    +
    \sum_{i=1}^n
    |\lambda_{\psi,i}(x_0)|
    |(\xi-\eta)^\top x_i| \\
    &\leq
    C_A
    \left\{
        1+
        \sum_{i=1}^n
        |\lambda_{\psi,i}(x_0)|
    \right\}
    \|\xi-\eta\| \\
    &\leq
    C_A(M+1)\|\xi-\eta\|.
    \label{eq:Q_equilipschitz}
\end{align}
Thus, $\{Q_n\}_{n\geq1}$ is equi-Lipschitz.

We now show (\ref{eq:Q_pointwise}). Suppose, to the contrary, that there exist
$\xi_0\in\operatorname{supp}(\nu)$, $\varepsilon>0$, and a subsequence,
still denoted by $n$, such that
$
    |Q_n(\xi_0)|\geq\varepsilon.
$
By \eqref{eq:Q_equilipschitz}, there exists $r>0$, independent of $n$,
such that
$
    |Q_n(\xi)|
    \geq
    {\varepsilon}/{2}
$
for
$
    \xi\in B(\xi_0,r).
$
Since $\xi_0\in\operatorname{supp}(\nu)$, we have
$
    \nu\{B(\xi_0,r)\}>0.
$
Consequently,
\[
    \int_{\mathbb R^p}
    |Q_n(\xi)|^2\nu(d\xi)
    \geq
    \frac{\varepsilon^2}{4}
    \nu\{B(\xi_0,r)\},
\]
which contradicts \eqref{eq:Q_L2_zero}.
Hence, \eqref{eq:Q_pointwise} holds.

\vspace{5pt}
\noindent
\textbf{(Step3)} In this step, we complete the proof.
Let $U$ be a nonempty open subset of
$\operatorname{supp}(\nu)$, whose existence is guaranteed by
Assumption~\ref{ass:miss}(iv). Recall that $C(A)$ is the space of
continuous real-valued functions on $A$. Extend $T_n$ complex-linearly
to $C(A;\mathbb C)$, and define
\[
    \mathcal E_U
    =
    \operatorname{span}_{\mathbb C}
    \left\{
        f_\xi:\xi\in U
    \right\},
    \qquad
    f_\xi(x)=e^{\mathrm{i}\xi^\top x}.
\]
For each $\xi\in U$, \eqref{eq:Q_pointwise} implies $T_n f_\xi-f_\xi(x_0)
    =
    -Q_n(\xi)
    \to0$. 
Therefore, by linearity of $T_n$,
\begin{equation}
    T_n f\to f(x_0)
    \qquad
    \text{for every } f\in\mathcal E_U.
    \label{eq:Tn_exp}
\end{equation}

Let $C(A;\mathbb C)$ denote the space of complex-valued continuous
functions on $A$. We claim that $\mathcal E_U$ is uniformly dense in
$C(A;\mathbb C)$.
To see this, suppose that it is not dense.
By the Hahn--Banach theorem and the Riesz representation theorem,
there exists a nonzero finite complex Borel measure $\mu$ supported
on $A$ such that $\int_A
    e^{\mathrm{i}\xi^\top x}
    \mu(dx)
    =
    0$, $\xi\in U$. 
Since $A$ is compact, the Fourier transform
$
    \widehat\mu(\xi)
    =
    \int_A
    e^{\mathrm{i}\xi^\top x}
    \mu(dx)
$
is real analytic on $\mathbb R^p$.
Because it vanishes on the nonempty open set $U$, it vanishes
identically on $\mathbb R^p$.
The uniqueness theorem for Fourier transforms of finite measures then
implies $\mu=0$, a contradiction.
Thus, $\mathcal E_U$ is uniformly dense in $C(A;\mathbb C)$.

Finally, let $g\in C(A)$ be real-valued and fix $\varepsilon>0$.
Since $C(A)$ is naturally embedded in $C(A;\mathbb C)$, the density
of $\mathcal E_U$ implies that there exists $f\in\mathcal E_U$ such that
$\sup_{x\in A}|g(x)-f(x)|<\varepsilon$. Therefore,
\begin{align}
    |T_ng-g(x_0)|
    &\leq
    |T_n(g-f)|
    +
    |T_nf-f(x_0)|
    +
    |f(x_0)-g(x_0)| \\
    &\leq
    (M+1)\varepsilon
    +
    |T_nf-f(x_0)|.
\end{align}
By \eqref{eq:Tn_exp}, $\limsup_{n\to\infty}
    |T_ng-g(x_0)|
    \leq
    (M+1)\varepsilon$.
Since $\varepsilon>0$ is arbitrary, we conclude that
\[
    T_ng\to g(x_0)
    \qquad
    \text{as }n\to\infty
\]
for every $g\in C(A)$.
\end{proof}

\subsection{Proof of Lemma~\ref{lem:average_2}}
\label{sec:average_2}

\begin{proof}
We prove the conclusion in several steps.

\vspace{5pt}
\noindent
\textbf{(Step1)}
For $x_0\in A$, define
\begin{equation}
    Q_{n,x_0}(\xi)
    =
    e^{\mathrm{i}\xi^\top x_0}
    -
    \sum_{i=1}^n
    \lambda_{\psi,i}(x_0)
    e^{\mathrm{i}\xi^\top x_i}.
\end{equation}
In this step, we show that 
\begin{equation}
    \sup_{x_0\in A}
    \int_{\mathbb R^p}
    |Q_{n,x_0}(\xi)|^2
    \nu(d\xi) \to 0.
    \label{eq:sup_conv}
\end{equation}
For each $x_0\in A$, let $i_n(x_0)
    \in
    \argmin_{1\leq i\leq n}
    \|x_i-x_0\|$ and define $h_n
    =
    \sup_{x_0\in A}
    \min_{1\leq i\leq n}
    \|x_i-x_0\|$.
As in the proof of Lemma~\ref{lem:average}, the spectral
representation and the optimality of $\lambda_\psi(x_0)$ give
\begin{align}
    \int_{\mathbb R^p}
    |Q_{n,x_0}(\xi)|^2
    \nu(d\xi)
    &=
    V\{\lambda_\psi(x_0);x_0\}
    \leq
    V\{e_{i_n(x_0)};x_0\} 
    =
    2\psi\left(
        \|x_{i_n(x_0)}-x_0\|
    \right).
\end{align}
Consequently,
\begin{equation}
    \sup_{x_0\in A}
    \int_{\mathbb R^p}
    |Q_{n,x_0}(\xi)|^2
    \nu(d\xi)
    \leq
    2\omega_\psi(h_n),
    \label{eq:uniform_Q_L2}
\end{equation}
where $\omega_\psi(r)
    =
    \sup_{0\leq u\leq r}\psi(u)$. 
As shown in the proof of Lemma~\ref{lem:average},
Assumption~\ref{ass:miss}(iv) implies that $\psi$ is continuous at
zero and $\psi(0)=0$. 
Since $h_n\to0$ by
\eqref{eq:dens_sample}, it follows that
$
    \omega_\psi(h_n)\to0.
$
Combining this together with \eqref{eq:uniform_Q_L2}, we have \eqref{eq:sup_conv}.

\vspace{5pt}
\noindent
\textbf{(Step2)}
In this step,  we show that, for every
$\xi\in\operatorname{supp}(\nu)$,
\begin{equation}
    \sup_{x_0\in A}
    |Q_{n,x_0}(\xi)|
    \to0.
    \label{eq:uniform_Q_pointwise}
\end{equation}
For that purpose, 
we first show that the family
$\{Q_{n,x_0}:n\geq1,\ x_0\in A\}$ is equi-Lipschitz.
Let $M
    =
    \sup_{n\geq1}
    \sup_{x_0\in A}
    \sum_{i=1}^n
    |\lambda_{\psi,i}(x_0)|
    <\infty$ and define $C_A
    =
    \sup_{x\in A}\|x\|
    <\infty$.
For every $x_0\in A$ and $\xi,\eta\in\mathbb R^p$,
\begin{align}
    |Q_{n,x_0}(\xi)-Q_{n,x_0}(\eta)|
    &\leq
    |(\xi-\eta)^\top x_0|
    +
    \sum_{i=1}^n
    |\lambda_{\psi,i}(x_0)|
    |(\xi-\eta)^\top x_i| \\
    &\leq
    C_A(M+1)\|\xi-\eta\|.
    \label{eq:uniform_Q_equilipschitz}
\end{align}
Thus, the family
$\{Q_{n,x_0}:n\geq1,\ x_0\in A\}$ is equi-Lipschitz.

Now we prove \eqref{eq:uniform_Q_pointwise} for every $\xi \in \operatorname{supp}(\nu)$.
Suppose, to the contrary, that there exist
$\xi_0\in\operatorname{supp}(\nu)$,
$\varepsilon>0$, a subsequence $\{n_k\}_{k\geq1}$ with
$n_k\to\infty$, and points
$x_{0,n_k}\in A$ such that
$
    |Q_{n_k,x_{0,n_k}}(\xi_0)|
    \geq
    \varepsilon
$
for all $k\geq1$.
By \eqref{eq:uniform_Q_equilipschitz}, there exists $r>0$,
independent of $k$, such that
$
    |Q_{n_k,x_{0,n_k}}(\xi)|
    \geq
    {\varepsilon}/{2}
$
for
$
    \xi\in B(\xi_0,r).
$
Since $\xi_0\in\operatorname{supp}(\nu)$, we have
$
    \nu\{B(\xi_0,r)\}>0.
$
Therefore,
\[
    \sup_{x_0\in A}
    \int_{\mathbb R^p}
    |Q_{n_k,x_0}(\xi)|^2
    \nu(d\xi)
    \geq
    \int_{\mathbb R^p}
    |Q_{n_k,x_{0,n_k}}(\xi)|^2
    \nu(d\xi)
    \geq
    \frac{\varepsilon^2}{4}
    \nu\{B(\xi_0,r)\}.
\]
The right-hand side is a strictly positive constant independent of
$k$, which contradicts \eqref{eq:sup_conv}.
Hence, \eqref{eq:uniform_Q_pointwise} holds.

\vspace{5pt}
\noindent
\textbf{(Step3)}
We complete the proof in this step. Let $U$ be a nonempty open subset of
$\operatorname{supp}(\nu)$, and extend $T_{n,x_0}$ to complex-valued
functions by the same formula. As in the proof of Lemma \ref{lem:average}, define
\[
    f_\xi(x)
    =
    e^{\mathrm{i}\xi^\top x},
    \qquad
    \mathcal E_U
    =
    \operatorname{span}_{\mathbb C}
    \{f_\xi:\xi\in U\}.
\]
For every $\xi\in U$, we have
\[
    \sup_{x_0\in A}
    |T_{n,x_0}f_\xi-f_\xi(x_0)|
    =
    \sup_{x_0\in A}
    |Q_{n,x_0}(\xi)|
    \to0.
\]
Every $f\in\mathcal E_U$ can be written as a finite linear
combination
$
f=\sum_{j=1}^m c_jf_{\xi_j}
$
with $\xi_j\in U$.
Hence, by linearity and the triangle inequality, for every $f\in\mathcal E_U$,
\begin{equation}
    \sup_{x_0\in A}
    |T_{n,x_0}f-f(x_0)|
    \to0.
    \label{eq:uniform_Tn_exp}
\end{equation}

Let $g\in C(A)$ and fix $\varepsilon>0$.
As shown in the proof of Lemma~\ref{lem:average},
$\mathcal E_U$ is uniformly dense in $C(A;\mathbb C)$.
Therefore, we can
choose $f\in\mathcal E_U$ such that $\sup_{x\in A}|g(x)-f(x)|<\varepsilon$.
Then
\begin{align}
    \sup_{x_0\in A}
    |T_{n,x_0}g-g(x_0)|
    &\leq
    \sup_{x_0\in A}
    |T_{n,x_0}(g-f)| 
    +
    \sup_{x_0\in A}
    |T_{n,x_0}f-f(x_0)|
    +
    \sup_{x_0\in A}
    |f(x_0)-g(x_0)| \\
    &\leq
    (M+1)\varepsilon
    +
    \sup_{x_0\in A}
    |T_{n,x_0}f-f(x_0)|.
\end{align}
By \eqref{eq:uniform_Tn_exp}, $\limsup_{n\to\infty}
    \sup_{x_0\in A}
    |T_{n,x_0}g-g(x_0)|
    \leq
    (M+1)\varepsilon$.
Since $\varepsilon>0$ is arbitrary, we conclude that
\[
    \sup_{x_0\in A}
    |T_{n,x_0}g-g(x_0)|
    \to0.
\]
\end{proof}

\section{Additional Real Data Analysis}

\subsection{Data Analysis for Exceedance Probabilities of Daily Average Temperatures}\label{sec:app_univariate_prob}

In this section, we use daily minimum and maximum
temperature records from the same 236 meteorological
stations as in Section \ref{sec:app_twodim_prob}
to predict the exceedance probability distributions of daily average temperatures. In contrast to Section \ref{sec:app_twodim_prob}, we use data from the entire year rather than restricting the analysis to June--September.
For each station $i=1,\ldots,236$, we observe pairs of daily minimum and maximum temperatures,
$
(Z_{i1},W_{i1}),\ldots,(Z_{im_i},W_{im_i}),
$
where $m_i$ denotes the number of pairs at station $i$.
Based on these observations, we construct proxies for the daily mean temperatures by
$
Y_{ij}
=
(Z_{ij}+W_{ij})/{2},
j=1,\ldots,m_i,
$
and define the exceedance probability function
\[
G_i(v)
=
\frac{1}{m_i}
\sum_{j=1}^{m_i}
I(Y_{ij}>v),
\qquad
v\in\mathbb{R}.
\]
Since $1-G_i$ is a cumulative distribution function, we regard each exceedance probability function as representing a univariate distribution and treat it as an element of the Wasserstein space (Example \ref{exm:unidist}).

We estimated the variogram following the procedure described in Section \ref{sec:implement}, setting $K=10$ and using the Mat\'ern correlation function \eqref{eq:matern}.
Figure~\ref{fig:app_dist_vario} shows the results. The fitted variogram model provides a good fit to the raw estimates. The estimated nugget variance and sill are $\hat{\tau}^2=0.216$ and $\hat{\tau}^2+\hat{\sigma}^2=2.382$, respectively. This indicates that the fitted working model assigns approximately 90.9\% of its sill to the spatial component.

\begin{figure}[t]
 \single
  \begin{minipage}[b]{0.49\linewidth}
    \centering
\includegraphics[width=0.9\columnwidth]{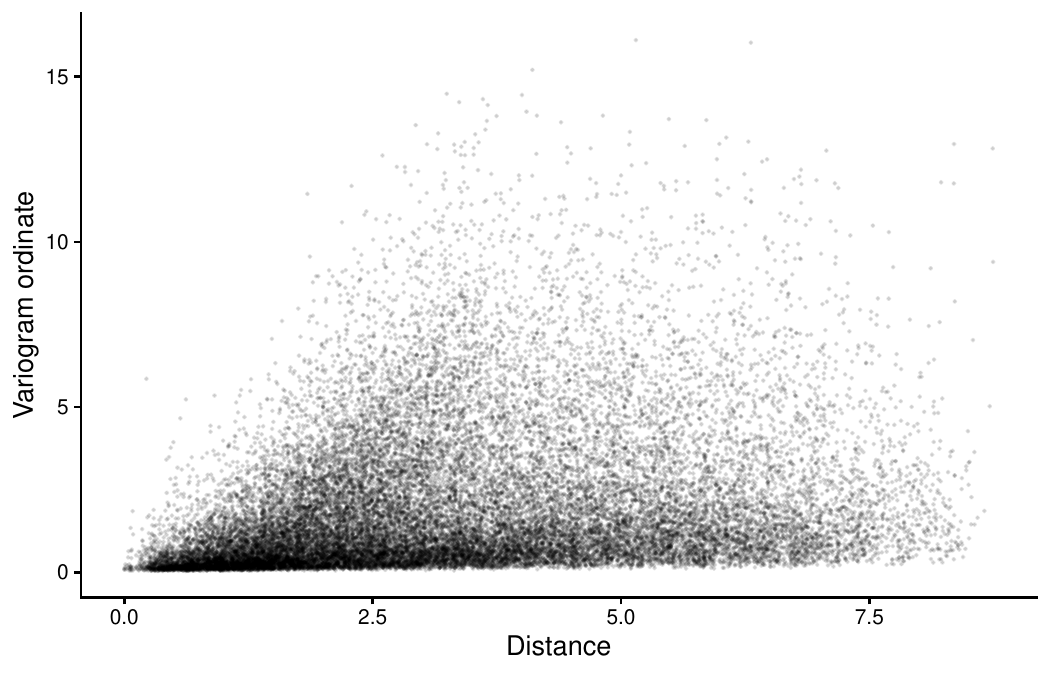}
    \subcaption{Empirical variogram}
  \end{minipage}
  \begin{minipage}[b]{0.49\linewidth}
    \centering
\includegraphics[width=0.9\columnwidth]{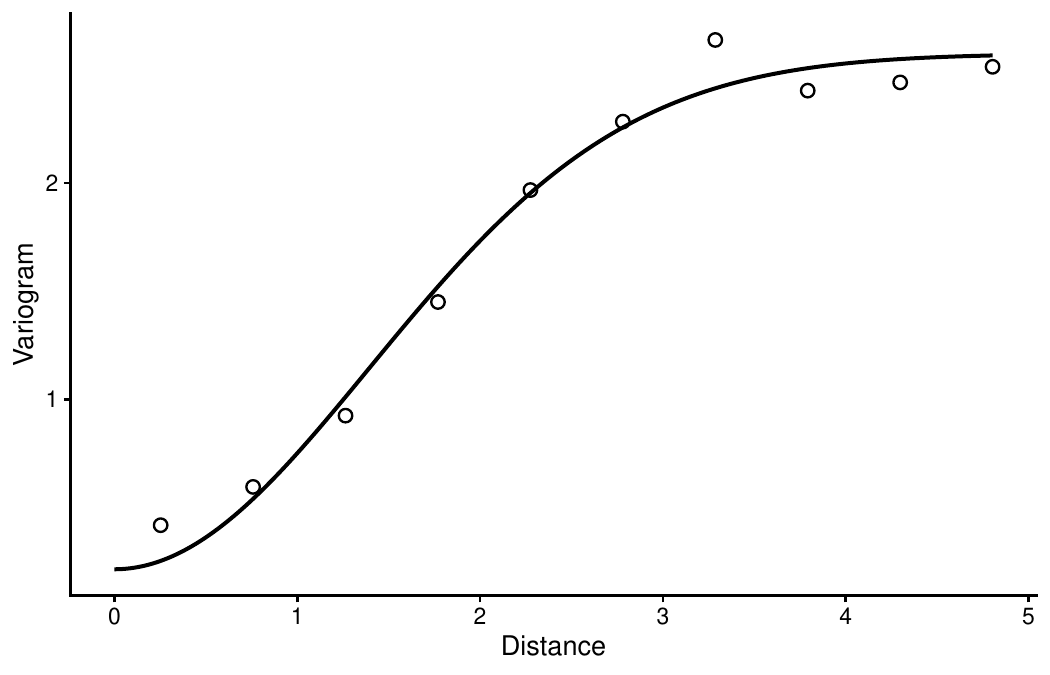}
    \subcaption{Mat\'ern variogram fit}
  \end{minipage}
  \caption{Results of variogram estimation for the exceedance probability distributions in Iowa. Panel (a) shows the empirical variogram based on Euclidean distances between longitude--latitude coordinates (in degrees) and the corresponding variogram ordinates. Panel (b) shows the raw variogram estimates (circles) together with the fitted Mat\'ern variogram model (solid line).}
  \label{fig:app_dist_vario}
\end{figure}

Next, we compare the predictive performance of the proposed method with
that of WPCA, LFR and GFR.
For each cross-validation fold, the variogram
parameters of the proposed method are estimated
using only the training data.
For LFR, the bandwidth parameters are selected by cross-validation
within each training fold. For WPCA, the minimum number of principal
components explaining more than 80\% of the total variance is retained.
Across the 236 held-out stations, the mean squared Wasserstein
prediction errors were 0.449 for Fr\'echet kriging,
0.516 for WPCA, 0.517 for LFR, and 0.691 for GFR.
Thus, Fr\'echet kriging achieved the lowest mean
cross-validation error, with reductions of approximately
13.0\%, 13.2\%, and 35.0\% relative to WPCA, LFR,
and GFR, respectively.

% \begin{table}[t]
%  \single
%     \centering
%     \begin{tabular}{cccc}
%        FK & WPCA & LFR&GFR \\ \hline  
%      \textbf{0.449}(0.601) &0.516(0.610) &0.517(0.616)&0.691(0.669) 
%     \end{tabular}
%    \caption{Summary of the leave-one-out cross-validation errors for the
% exceedance probability distributions in Iowa, based on squared
% Wasserstein distance over 236 held-out stations.
% Values are mean site-level prediction errors, with standard deviations
% of the site-level errors in parentheses.
% FK denotes the proposed Fr\'echet kriging method, WPCA denotes kriging
% based on Wasserstein principal component analysis \citep{qiu2026spatial},
% and LFR and GFR denote local and global Fr\'echet regression,
% respectively \citep{petersen2019frechet}.}
%     \label{tab:app_dist_1}
% \end{table}

Finally, we partition the spatial domain defined by longitudes [-97.5, -89.0] and latitudes [40.0, 44.0] into a $60\times30$ grid and predict the exceedance probability distribution at the center of each grid cell using the proposed method. Figure \ref{fig:mesh_1} displays the predicted exceedance probabilities $P(Y>v)$ for $v=0$, $10$, and $20$. 
The predicted maps vary smoothly over the spatial domain, reflecting the spatial dependence captured by the proposed model. As the temperature thresholds increase, the predicted exceedance probabilities decrease uniformly across the region.

\begin{figure}
\single
    \centering
    \includegraphics[width=0.99\linewidth]{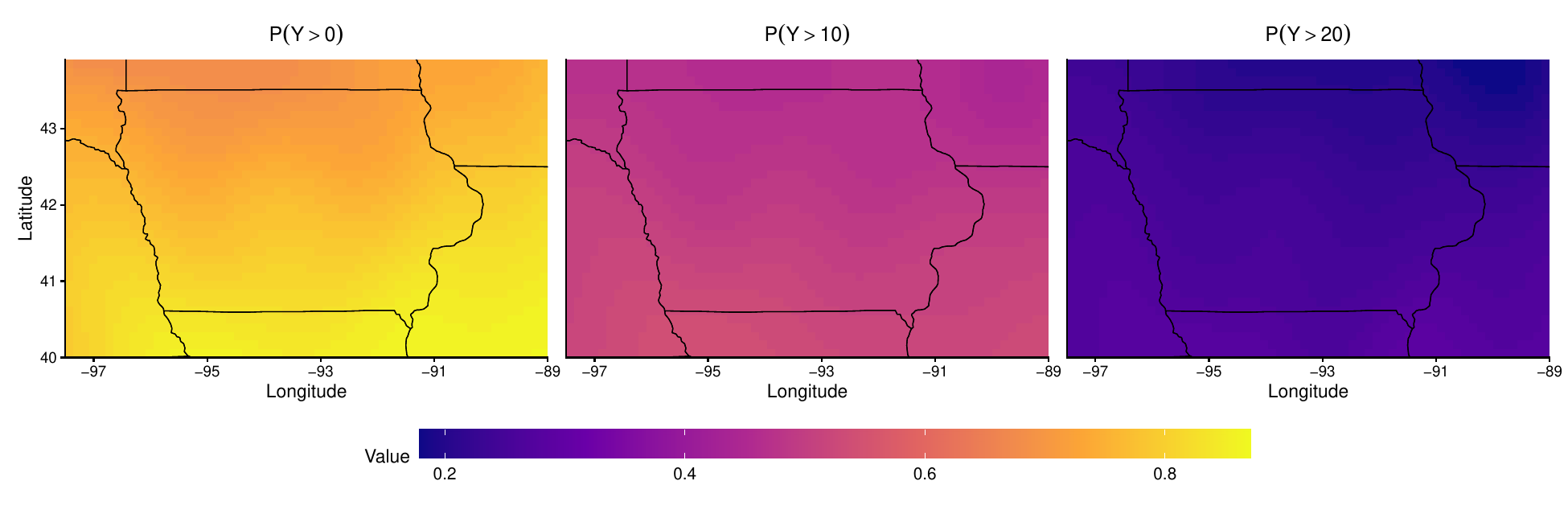}
    \caption{Predicted exceedance probability maps by the proposed method, $P(Y>v)$, for $v=0$ (left), $10$ (center), and $20$ (right).}
    \label{fig:mesh_1}
\end{figure}

\subsection{Additional Figures}
Figure \ref{fig:app_maps_1} shows the locations of the meteorological stations in the Iowa temperature dataset.
Figure \ref{fig:Germany} shows the locations of regions in the network dataset.

For the application in Section~\ref{sec:app_twodim_prob},
Figure~\ref{fig:one_shot_twodim} displays heatmaps of the
observed density functions alongside their Fr\'echet kriging
predictions at two stations.
These stations are selected such that their cross-validation
errors are closest to the first and third quartiles of the
errors across all stations.
In both cases, the predicted densities closely resemble
the observed densities in terms of their overall location,
spread, and shape.

For the application in Section~\ref{sec:app_univariate_prob},
Figure~\ref{fig:one_shot_unidist} shows the observed quantile
functions alongside their Fr\'echet kriging predictions at two stations.
As in the previous application, these stations are selected such that
their cross-validation errors are closest to the first and third
quartiles of the errors across all stations.
At both stations, the predicted quantile functions closely follow
the observed ones over almost the entire probability range,
indicating that the proposed method effectively captures
the overall shapes of the temperature distributions.

\begin{figure}[tb]
 \single
  \begin{minipage}[b]{0.49\linewidth}
    \centering
\includegraphics[width=0.9\columnwidth]{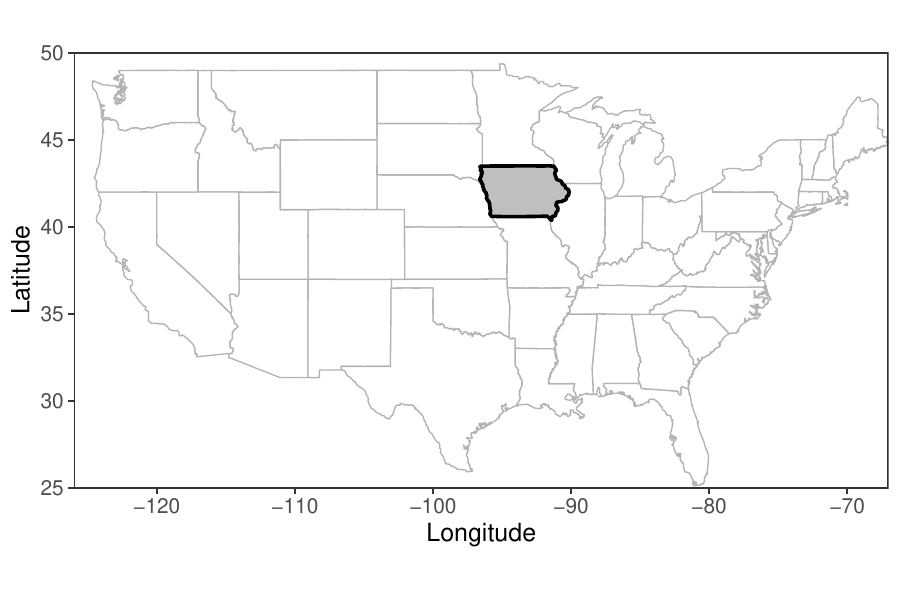}
    \subcaption{Iowa}
  \end{minipage}
  \begin{minipage}[b]{0.49\linewidth}
    \centering
\includegraphics[width=0.9\columnwidth]{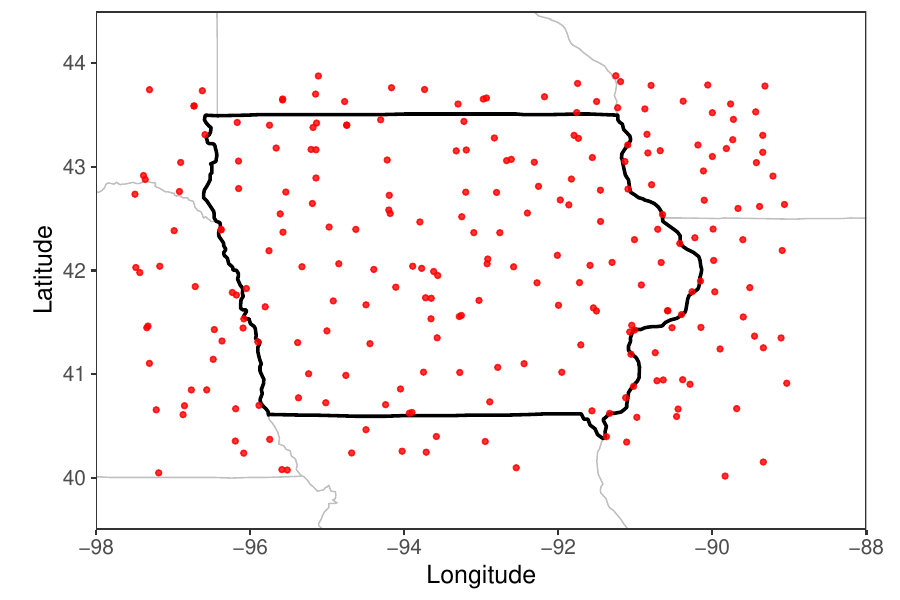}
    \subcaption{Station locations}
  \end{minipage}
  \caption{Locations of the meteorological stations in the Iowa temperature dataset. Panel (a) shows the location of Iowa within the United States, and panel (b) shows the locations of the 236 meteorological stations.}
  \label{fig:app_maps_1}
\end{figure}

\begin{figure}
    \centering
    \includegraphics[width=0.5\linewidth]{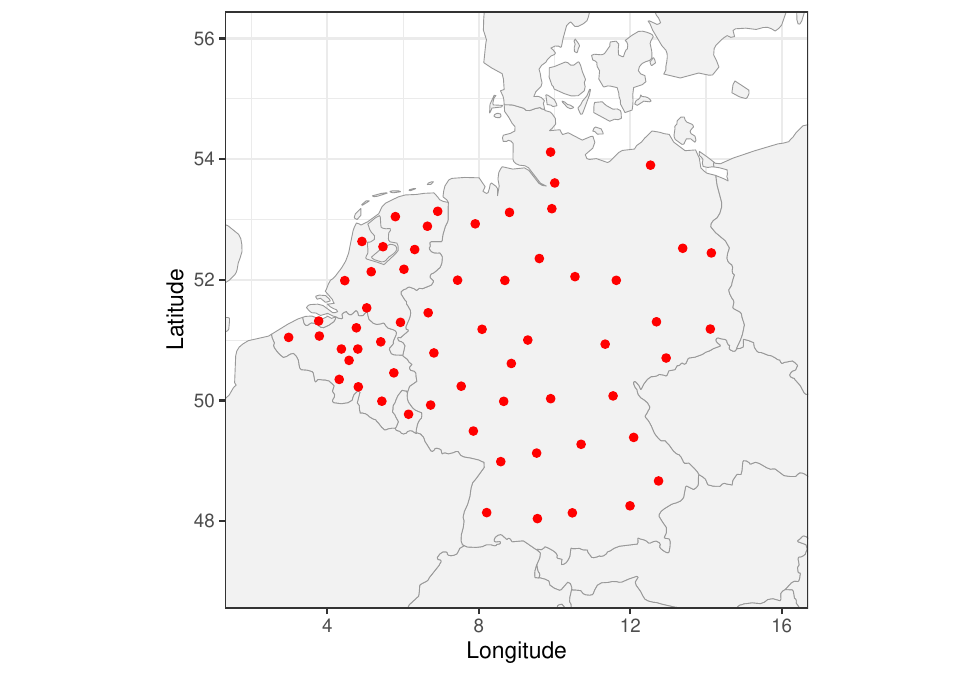}
    \caption{61 regions in the network dataset.}
    \label{fig:Germany}
\end{figure}

\begin{figure}[t]
\single
    \centering
    \includegraphics[width=0.99\linewidth]{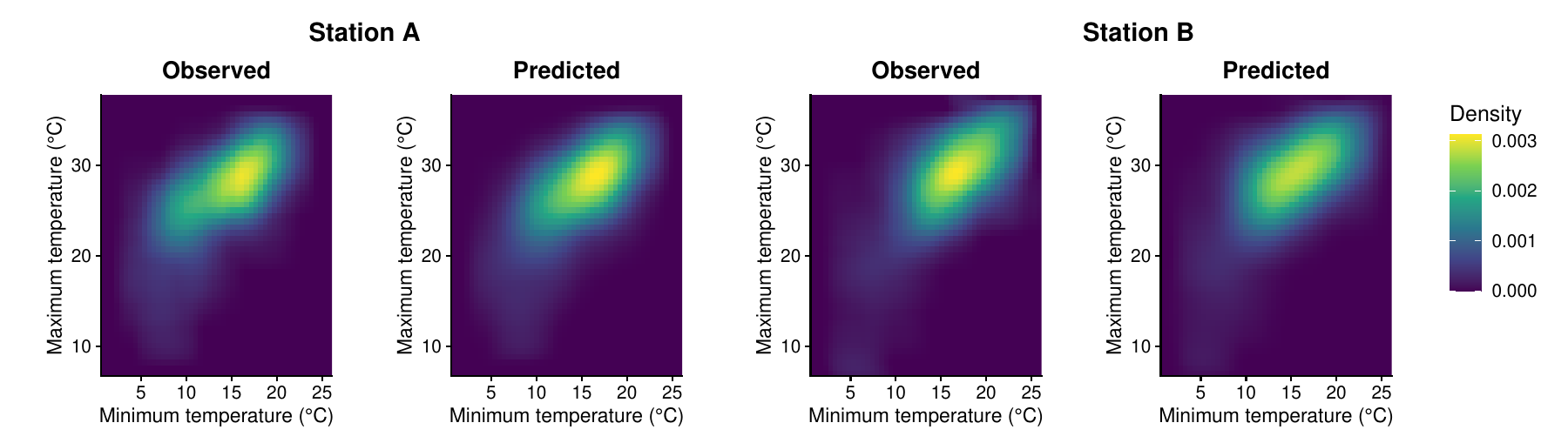}
    \caption{Heatmaps of the observed density functions and those predicted by the proposed method at two stations. Stations A and B are located at (longitude, latitude) $=(-91.45, 42.47)$ and $=(-96.73, 43.59)$, respectively. The squared Fisher--Rao distances between the observed and predicted distributions are 0.0140 and 0.0344 at Stations A and B, respectively.}
    \label{fig:one_shot_twodim}
\end{figure}

\begin{figure}[t]
\single
    \centering
    \includegraphics[width=0.8\linewidth]{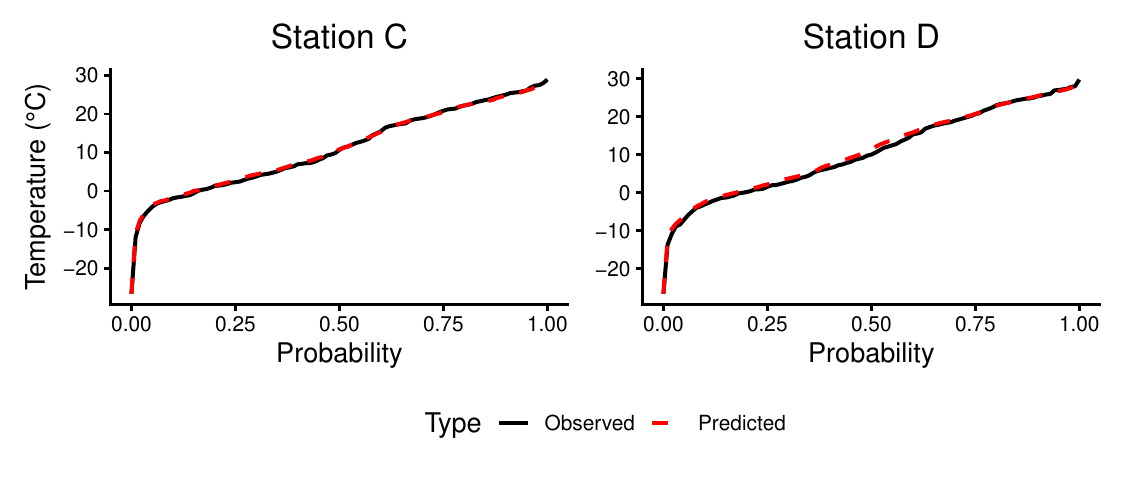}
    \caption{Observed (black solid lines) and predicted (red dashed lines) quantile functions at two stations in the Iowa temperature dataset. Stations C and D are located at (longitude, latitude) $=(-89.03, 40.91)$ and $=(-96.19, 40.67)$, respectively. The squared Wasserstein distances between the observed and predicted distributions are 0.135 and 0.520 at Stations C and D, respectively.}
    \label{fig:one_shot_unidist}
\end{figure}

\end{document}